\documentclass{article}

\usepackage[margin=1.2in]{geometry}
\usepackage[utf8]{inputenc} 
\usepackage[T1]{fontenc}    
\usepackage{lmodern}
\usepackage{amsmath,amsfonts,amsthm,amssymb}
\usepackage{mathabx}
\usepackage{algorithm}
\usepackage{algpseudocode}
\usepackage[numbers]{natbib}

\usepackage{url}            
\usepackage{booktabs}       
\usepackage{numprint}

\usepackage{hyperref}  
\usepackage{nicefrac}       
\usepackage{microtype}      
\usepackage{mathtools}
\usepackage{xfrac}
\usepackage{csquotes}
\usepackage{multirow}
\usepackage{multicol}

\usepackage{bbm}

\usepackage[]{xcolor}

\newcommand{\ptrue}{p^{\mathrm{lim}}}

\newcommand{\pbc}{P^{\mathrm{BC}}}

\newcommand{\pperm}{P^{\mathrm{perm}}}
\newcommand{\plim}{P_{\mathrm{lim}}}

\newcommand{\pval}{P}
\newcommand{\bd}{\boldsymbol{d}}
\newcommand{\Wbm}{\bar{W}}

\graphicspath{{fig/}}

\usepackage{graphicx} 
\usepackage{caption}
\usepackage{subcaption}
\usepackage{enumitem}
\usepackage{cleveref}

\newtheorem{theorem}{Theorem}[section]

\newtheorem{proposition}[theorem]{Proposition}
\newtheorem{lemma}[theorem]{Lemma}
\newtheorem{assumption}{Assumption}
\newtheorem{property}{Property}

\newtheorem{example}{Example}

\newif\ifarxiv

\arxivtrue

\title{Multiple testing with anytime-valid Monte Carlo p-values}

 \author{%
  Lasse Fischer\thanks{University of Bremen, Germany. Email: \texttt{fischer1@uni-bremen.de}} \and Timothy Barry\thanks{Harvard University, USA.  Email: \texttt{tbarry@hsph.harvard.edu} } \and Aaditya Ramdas\thanks{Carnegie Mellon University, USA. Email: \texttt{aramdas@cmu.edu} }
}

\begin{document}

\maketitle
\begin{abstract}
In contemporary problems involving  genetic or neuroimaging data,  thousands of hypotheses need to be tested. Due to their high power, and finite sample guarantees on type-I error under weak assumptions, Monte Carlo permutation tests are often considered as gold standard for these settings. However, the enormous computational effort required for (thousands of) permutation tests is a major burden. 
In this paper, we integrate recently constructed anytime-valid permutation p-values into a broad class of multiple testing procedures, including the Benjamini-Hochberg procedure. This allows to fully adapt the number of permutations to the underlying data and thus, for example, to the number of rejections made by the multiple testing procedure. Even though this data-adaptive stopping can induce dependencies between the p-values that violate the usual assumptions of the Benjamini-Hochberg procedure, we prove that our approach controls the false discovery rate under mild assumptions. Furthermore, our method provably decreases the required number of permutations substantially without compromising power. On a real genomics data set, our method reduced the computational time from more than three days to less than four minutes while increasing the number of rejections.

\end{abstract}

\tableofcontents

\section{Introduction}

In a classical Monte Carlo permutation test, we observe some test statistic $Y_0$, generate $B$ additional test statistics $Y_1,\ldots, Y_B$ and calculate the permutation p-value by 
\begin{align}\pperm_B=\frac{1+\sum_{t=1}^B \mathbbm{1}\{Y_t\geq Y_0\}}{1+B}.\label{eq:pperm}\end{align}
  In order to be able to reject the null hypothesis at level $\alpha$, we would need to generate at least $B\geq 1/\alpha -1$ permuted datasets with corresponding test statistics. This can already lead to considerable computational effort when the data generation process is demanding. However, it is significantly higher when $M$ hypotheses are tested instead of a single one. First, we need to perform a permutation test for each hypothesis, leading to a minimum number of $M(1/\alpha -1)$ permutations in total. Second, when multiple testing corrections are performed, each hypothesis is tested at a lower individual level than the overall level $\alpha$. For example, if the Bonferroni correction is used, each hypothesis is tested at level $\alpha/M$ leading to the requirement of at least $M(M/\alpha -1)\approx  M^2/\alpha$ permutations. Even if we use more powerful multiple testing procedures such as the Benjamini-Hochberg (BH) procedure \citep{benjamini1995controlling}, this lower bound for the total number of permutations usually remains. The problem is that $B$ needs to be chosen in a manner that protects against the worst case in which all p-values are large except for one that is then essentially tested at level $\alpha/M$. 
  
  In this paper, we consider applying multiple testing procedures to anytime-valid permutation p-values which allows to stop resampling early and reject (or accept) a hypothesis as soon as the corresponding p-value is rejected (or accepted) by the multiple testing procedure. In this way, we can adapt the number of permutations to the number of rejections and potentially save a lot of permutations. In Figure~\ref{fig:worst_bounds}, we compare the required number of permutations to ensure that enough samples are drawn in the worst case when using BH with the permutation p-value and our adaptive method, showing that our method reduces the linear bound $M/\alpha -1$ to a logarithmic one.

  \begin{figure}[h!]
\centering
\includegraphics[width=0.7\textwidth]{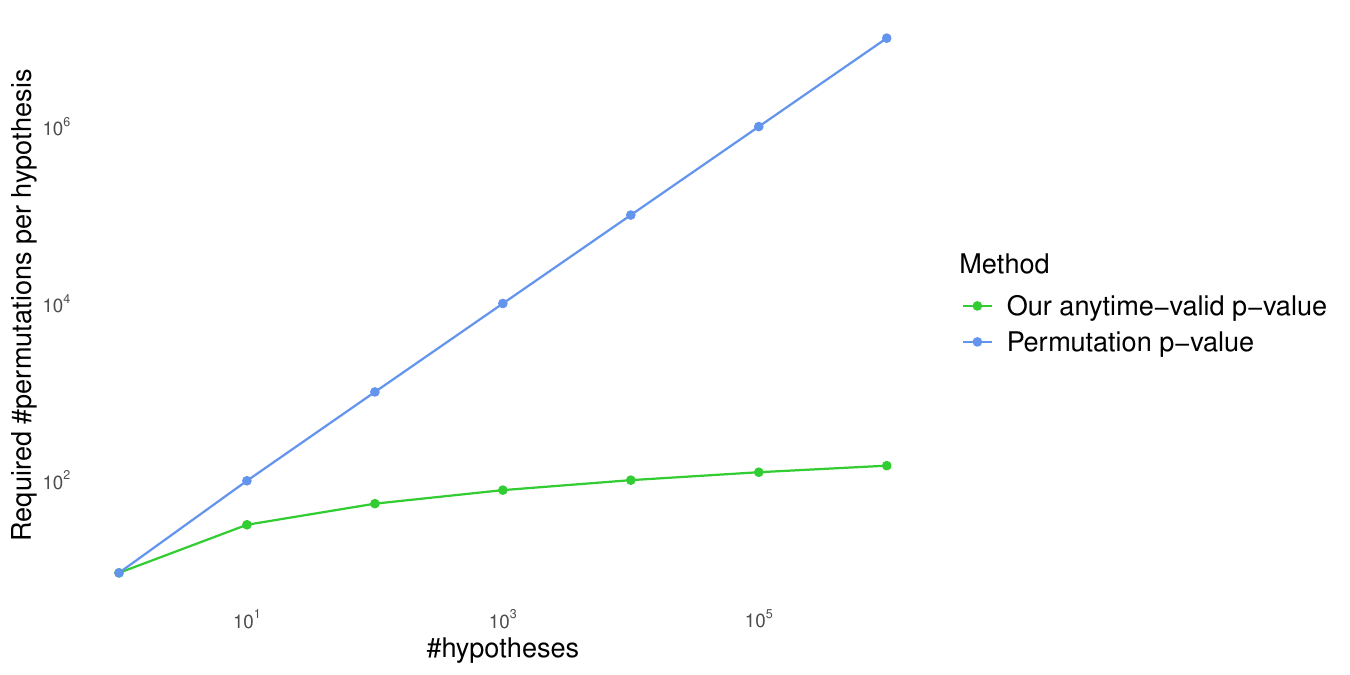}
\caption{Required number of permutations per hypothesis to ensure that enough samples are drawn in the worst case when using the BH procedure and our anytime-valid p-value ($\alpha=0.1$). While we need to draw at least $M/\alpha-1$ per hypothesis with the permutation p-value, we show that the average number of samples drawn per hypothesis scales logarithmically with $M$ (in the worst case) when using our adaptive method (see \eqref{eq:worst_bound} for $h=1$).  \label{fig:worst_bounds} }\end{figure}

Despite our terminology and notation focusing on \textit{permutation} tests, everything in this paper identically applies to other Monte Carlo methods such as \textit{conditional randomization tests} as well, which are also common in the causal inference literature, and have been applied to genetic settings under the so-called model-X assumption \citep{candes2018panning, bates2020causal, barry2021sceptre}.

There are many existing approaches considering sequential Monte Carlo tests for multiple testing --- in particular for the BH procedure \citep{jiang2012statistical, guo2008adaptive, zhang2019adaptive, sandve2011sequential, 
bancroft2013estimation,
gandy2014mmctest, gandy2016framework, gandy2017quickmmctest, pounds2011integrated}. These works can be broadly categorized into the following three approaches.
\begin{enumerate}
    \item The first idea is to use early stopping rules that ensure obtaining the same decisions as applying the multiple testing procedure to the permutation p-value $\pperm_B$ for a fixed $B\in \mathbb{N}$ with high probability \citep{jiang2012statistical, guo2008adaptive, zhang2019adaptive}.
    \item Another line of work provides bounds for the probability of obtaining a different decision than the limiting permutation p-value \citep{gandy2014mmctest, gandy2016framework, gandy2017quickmmctest}. Here, the main goal is not to reduce the number of permutations but to minimize the resampling risk \citep{fay2002designing}.
    \item In a third approach  \citep{sandve2011sequential, pounds2011integrated,bancroft2013estimation},  multiple testing procedures are applied to the sequential permutation p-value by \citet{besag1991sequential}.
\end{enumerate}

However, existing works in all three directions are limited. The first approach is conservative by nature, since it splits the error budget between the probability of obtaining the same decision as $\pperm_B$ and the type I error probability of $\pperm_B$ itself \citep{zhang2019adaptive}; the second approach does not save computational resources; and the third approach allows for early stopping under the null hypothesis only. To the best of our knowledge, the only existing method allowing to adapt the number of permutations to the number of rejections while guaranteeing false discovery rate (FDR) control, is the AMT algorithm by \citet{zhang2019adaptive}. They fix the maximum number of permutations $B\in \mathbb{N}$ in advance and consider the  generation process of the test statistics as a multi-armed bandit  problem, and their guarantee falls into the first category above.


\ifarxiv
\paragraph{Our contributions.}
\fi
We will follow a strategy similar to the third path, however, instead of the Besag-Clifford p-value, we consider the anytime-valid permutation p-values recently constructed by \citet{fischer2024sequential}. 
The latter framework encompasses the former as a special case, but is able to  reduce the number of permutations under both the null and the alternative, without compromising power or type-I error. With this, our approach allows to adapt the number of permutations to the number of rejections, which reduces the required computational time substantially, while providing provable error control. 

We particularly study the use of our recommended ``anytime-valid BC'' p-value with the BH procedure, yielding the following desirable results. 
\begin{enumerate}
    \item FDR control under PRDS of the limiting permutation p-values (Theorem~\ref{theo:PRDS}).
    \item Equivalence to the permutation p-value \eqref{eq:pperm} applied with a particular data-adaptive number of permutations $B$, showing that no power is impaired when using our anytime-valid method (Theorem~\ref{theo:magic}). 
    \item Upper bound for the average number of permutations required per hypothesis that scales logarithmically with the number of hypotheses $M$ (Proposition~\ref{prop:upper_bound}).
    \item  Substantial reduction of computational effort and good power in several simulated and real data experiments. For example, on a real genomics data set our method reduced the computational time compared to the classical permutation p-value \eqref{eq:pperm} from more than three days to about three minutes and thirty seconds while making more discoveries (Section~\ref{sec:real_data}).
\end{enumerate}

Furthermore, we provide results on error rate control for arbitrary multiple testing procedures and anytime-valid p-values. If the employed multiple testing procedure (which takes p-values as input and produces a rejection set as output) provides the desired error rate control under arbitrary dependence of the p-values, then it follows that it retains such validity with anytime-valid permutation p-values (Theorem~\ref{theo:FDR_arbitrary_dep}), allowing us to stop fully data-adaptively with no danger of error rate control. If a specific dependence structure is required to control the desired error rate, then some care regarding the stopping time and employed anytime-valid permutation p-value needs to be taken. For brevity, we restrict to the BH procedure in this case and prove FDR control under independent limiting permutation p-values (Proposition~\ref{prop:futility_ind}).

\ifarxiv
\paragraph{Paper outline.}
\fi
 In Section~\ref{sec:old_paper}, we recap the notion of anytime-valid permutation p-values and the anytime-valid BC method as particular instance. Afterward, we show how this general methodology can be used in multiple testing (Section~\ref{sec:con_cal}). The proposed technique is straightforward for all p-value based multiple testing procedures that do not require assumptions on the dependence structure. In Section~\ref{sec:bh}, we discuss the use of the BH procedure with the anytime-valid BC p-value; proving FDR control under weak assumptions, showing that no power is impaired compared to the classical approach and providing an upper bound for the number of permutations. 
In Section~\ref{sec:anytime_betting}, we show how the general testing by betting approach to construct anytime-valid permutation p-values, introduced by \citet{fischer2024sequential}, can be used in multiple testing.  
Finally, we demonstrate the application of the BH procedure with anytime-valid permutation p-values to simulated data and real data in Sections~\ref{sec:sim} and~\ref{sec:real_data}, respectively. All proofs are provided in Supplementary Material~\ref{appn:proofs}.


\section{Anytime-valid permutation p-values\label{sec:old_paper}}


In this section, we recap the notion of an anytime-valid permutation p-value \citep{fischer2024sequential, johari2022always, howard2021time} and introduce the particular instance that we will mainly focus on throughout this paper. 

Suppose we observed some data $X_0$ and have the ability to generate additional data $X_1,X_2,\ldots$, e.g., by permuting the treatment labels of $X_0$ in a treatment vs.\ control trial. Let $Y_i=S(X_i)$ for some test statistic $S$. We consider testing the null hypothesis 
\begin{align}H_0: Y_0,Y_1, \ldots \text{ are exchangeable,} \label{eq:H_0}\end{align}
where \textit{exchangeable} means that the joint distribution of the sequence does not change for any finite permutation of the test statistics. 

Note that this hypothesis implicitly assumes that we sample the permutations with replacement, since there would be an upper bound for the number of permutations otherwise. However, this is the usual state of affairs in practice, since remembering which permutations have already been drawn is computationally intensive and if the data size is large, not all permutations can be drawn anyway \citep{stark2018random}. Also, as implicitly contained in the formulation of the permutation p-value \eqref{eq:pperm}, we usually assume that larger values of $Y_0$ provide more evidence against the null hypothesis. However, all methods apply equally to left-tailed problems; and two-tailed tests can be obtained by combining a right-tailed and left-tailed test. 

The standard approach to this testing problem is the permutation p-value \eqref{eq:pperm}. In general, one would want to choose $B$ as large as possible, since the power increases with $B$. However, the computational effort can be too high for large $B$. In order to reduce the computational cost, we consider generating the test statistics $Y_1,Y_2,\ldots$ sequentially in this paper. At each step $t$, a test statistic $Y_t$ is generated, compared with $Y_0$, and then a p-value $\pval_t$ is calculated. The resulting process $(\pval_t)_{t\in \mathbb{N}}$ is called \emph{anytime-valid p-value} \citep{johari2022always, howard2021time}, if 
\begin{align}
    \mathbb{P}_{H_0}( \pval_{\tau} \leq \alpha)\leq \alpha  \text{ for all } \alpha \in [0,1], \label{eq:anytime_p-value_condition}
\end{align}
 where $\tau$ is an arbitrary stopping time. Hence, an anytime-valid p-value allows to stop the sampling process at any time, while providing valid type I error control. In particular, it is usually stopped as soon as a hypothesis can be rejected. However, if the evidence against the null hypothesis is weak, it can also be sensible to stop for futility (stop and accept $H_0$). We note two properties \citep{howard2021time} of anytime-valid p-values that will be important for us when applying them with multiple testing procedures.
\begin{enumerate}[label=(\Alph*)]
    \item If $(\pval_t)_{t\in \mathbb{N}}$ is an anytime-valid p-value, then $(\tilde{\pval}_t)_{t\in \mathbb{N}}$, where $\tilde{\pval}_t=\min_{s\leq t} \pval_s$, is an anytime-valid p-value as well. Hence, it can be assumed that all anytime-valid p-values $(\pval_t)_{t\in \mathbb{N}}$ are \emph{nonincreasing} in $t$. \label{bull:A}
    \item If $(\pval_t)_{t\in \mathbb{N}}$ is an anytime-valid p-value, it holds that $\mathbb{P}_{H_0}( \pval_{T} \leq \alpha)\leq \alpha$ for any data-dependent \emph{random time} $T$. This means that the stopping time $\tau$ for our anytime-valid p-value does not have to  be adapted to the filtration generated by $Y_t, t\geq 0$, but can depend on any (possibly data-dependent) information. In particular, this allows to stop based on the data of other hypotheses. We will still denote $\tau$ as stopping time in the following. \label{bull:B}
\end{enumerate}

In what follows, we present a simple --- but, as we shall see, powerful --- anytime-valid generalization of the Besag-Clifford p-value.
We will later (Section~\ref{sec:anytime_betting}) also introduce a general way to construct anytime-valid permutation p-values using the betting framework introduced by \citet{fischer2024sequential}. However, due to its simplicity, desirable behavior (Section~\ref{sec:bh}) and remarkable performance in our experiments (Sections~\ref{sec:sim} and \ref{sec:real_data}), the anytime-valid Besag-Clifford p-value is our proposed choice and will therefore be examined in more detail.


\ifarxiv
\paragraph{The anytime-valid Besag-Clifford p-value.}
\fi
\citet{besag1991sequential} introduced a \emph{sequential} permutation p-value:
\begin{align}\pbc_{\gamma(h,B)}=\begin{cases}
            h/\gamma(h,B), & L_{\gamma(h,B)}=h \\
            (L_{\gamma(h,B)}+1)/(B+1), & \text{otherwise},
        \end{cases}\label{eq:Besag-Clifford}\end{align}
where $h$ is some predefined parameter, $L_B:=\sum_{t=1}^B \mathbbm{1}\{Y_t\geq Y_0\}$ is the random number of ``losses'' after $B$ permutations and $\gamma(h,B)=\min( \inf\{t\in \mathbb{N}:L_t=h\}, B)$ is a stopping time. In case of $B=\infty$, we just write $\gamma(h)$. The Besag-Clifford (BC) p-value only allows to stop at the particular time  $\gamma(h,B)$ and not at any other stopping time. Hence, it allows to stop early when  $h$ losses occurred before all $B$ permutations were sampled and therefore only saves computations under the null hypothesis, but not when the evidence against the null is strong. For this reason, it is not well-suited to adapt the number of permutations to the number of rejections when multiple hypotheses are considered. However, the BC p-value permits a simple anytime-valid  generalization which solves this issue.

For simplicity, assume $B=\infty$, meaning the Besag-Clifford method would only stop after observing $h$ losses (the case $B<\infty$ is also possible \citep{fischer2024sequential}, however, we do not gain much and the notation is more complicated). Suppose we plan to use the BC p-value $\pbc_{\gamma(h)}$, are currently at some time $t<\gamma(h)$, and since the evidence against $H_0$ already seems very strong, we would like to stop the process to save computational time. What (non-trivial) p-value could we report at this time while ensuring \eqref{eq:anytime_p-value_condition} is fulfilled? The answer is very simple, just report the p-value under the most conservative future, meaning if all upcoming test statistics are losses. Since we need to sample $h-L_{t}$ test statistics until we hit the stopping time $\gamma(h)$, this leads to the following generalization of the Besag-Clifford p-value
\begin{align}\pval_{t}^{\mathrm{avBC}}=\begin{cases}
            \frac{h}{t+h-L_{t}}, & t\leq \gamma(h) \\
            h/\gamma(h), & t> \gamma(h).
        \end{cases}
 \label{eq:avBC} \end{align}
Since $\pval_{\gamma(h)}^{\mathrm{avBC}}=\pbc_{\gamma(h)}$, its anytime-validity \eqref{eq:anytime_p-value_condition} follows immediately by the validity of the Besag-Clifford p-value at time $\gamma(h)$,
$$
\mathbb{P}_{H_0}(\pval_{\tau}^{\mathrm{avBC}}\leq \alpha) \leq \mathbb{P}_{H_0}(\pval_{\gamma(h)}^{\mathrm{avBC}}\leq \alpha)=\mathbb{P}_{H_0}(\pbc_{\gamma(h)}\leq \alpha) \leq \alpha,
$$
where $\tau$ is an arbitrary stopping time.
Apart from \citet{fischer2024sequential} mentioning the anytime-valid BC method as special instance of their general approach (see Section~\ref{sec:anytime_betting}), we have not seen it being exploited in the literature before. The simple idea of stopping early by assuming a conservative future was previously explored for the classical permutation p-value, but not found to be very useful \citep{hapfelmeier2023efficient, good2013permutation}. However, it is much more efficient for the BC method and particularly useful in the multiple testing setting, because it allows to adapt the stopping time to the individual significance level (see Section~\ref{sec:magic} for a general treatment; for now, we just illustrate this behavior via a simple example). 

\begin{example} Suppose $h=10$ and $B=999$. If $\alpha=0.01$, the anytime-valid generalizations of the BC method and the permutation p-value are equivalent. They can both stop for accepting the null hypothesis as soon as the number of losses equals $10$ or stop for rejecting the null if the number of losses after $990+\ell$ permutations is $\ell$, $\ell\in \{0,1,\ldots,9\}$. Now suppose $\alpha$ increases to $0.05$. Suddenly, the anytime-valid BC method allows to stop for rejection after $190+\ell$ permutations, if $\ell$ losses are observed, while at least $950$ permutations would need to be sampled such that a rejection can be obtained by the permutation p-value.
\end{example}

In case of $h=1$, \citet{fischer2024sequential} referred to the anytime-valid BC p-value as the \emph{aggressive strategy}, as it already stops after the first loss. It also requires the least possible number of permutations, which can be useful in scenarios with enormous computational demand, and therefore deserves a special mention.

\section{Multiple testing with anytime-valid permutation p-values\label{sec:con_cal}}

Suppose we have $M$ null hypotheses $H_0^1,\ldots, H_0^M$ of the form \eqref{eq:H_0}. 
We denote by $Y_0^i, Y_1^i,\ldots $ the sequence of test statistics for $H_0^i$ and by $L_B^i$ the random number of losses for hypothesis $H_0^i$ after $B$ permutations. Furthermore, let \[I_0 \subseteq \{1,\ldots,M\}\] be the index set of true hypotheses, $R\subseteq\{1,\ldots,M\}$ the set of rejected hypotheses and $V=I_0\cap R$ the set of falsely rejected hypotheses. Two of the most common error rates considered in multiple testing are the familywise error rate (FWER) and the FDR. The FWER is defined as probability of rejecting any true null hypothesis
\begin{align}
    \mathrm{FWER}:=\mathbb{P}(|V|>0). \label{eq:FWER}
\end{align}
The FDR is the expected proportion of falsely rejected hypotheses among all rejections 
\begin{align}
\mathrm{FDR}:=\mathbb{E}\left[\frac{|V|}{|R|\lor 1}\right]. \label{eq:FDR}
\end{align}
The aim is to control either the FWER or FDR at some prespecified level $\alpha\in (0,1)$. Although we focus on FWER and FDR in this paper for concreteness, our approach also works with all other common error metrics (false discovery proportion tail probabilities, k-FWER, etc.).
Control of the FWER implies control of the FDR \citep{benjamini1995controlling} and is therefore more conservative. Hence, FDR controlling procedures often lead to more rejections and are more appropriate in large-scale exploratory hypothesis testing. In non-exploratory validation studies, FWER control is usually the norm.


\ifarxiv
\paragraph{Our general multiple testing approach.}
\fi
Our general algorithm requires two ingredients, an anytime-valid permutation p-value $(\pval_t^i)_{t\in \mathbb{N}}$ \eqref{eq:anytime_p-value_condition}, like the anytime-valid BC p-value \eqref{eq:avBC}, for each hypothesis $H_0^i$, $i\in \{1,\ldots,M\}$, and a monotone multiple testing procedure $\bd=(d_i)_{i\in \{1,\ldots,M\}}$. A multiple testing procedure $\bd$ is called monotone, if each $d_i:[0,1]^M\to \{0,1\}$, which maps p-values for $H_0^1,\ldots, H_0^M$ to the decision for $H_0^i$ ($1=\text{reject; } 0=\text{accept}$), is nonincreasing in all coordinates. Hence, if some of the p-values become smaller, all previous rejections remain and additional rejections might be obtained. 

Recall the two properties of anytime-valid p-values from Section~\ref{sec:old_paper}:~\ref{bull:A},  $(\pval_t^i)_{t\in \mathbb{N}}$ is nonincreasing in $t$, and~\ref{bull:B}, $(\pval_t^i)_{t\in \mathbb{N}}$ is valid at all data-dependent random times. Due to~\ref{bull:B}, we can stop the sampling process for $H_0^i$ based on all available data, and particularly as soon as $H_0^i$ can be rejected by $\bd$, without violating the validity of $(\pval_t^i)_{t\in \mathbb{N}}$. In addition, due to~\ref{bull:A} and the monotonicity of $\bd$, we can already report the rejection of $H_0^i$ at that time, as the decision won't change if we continue testing. Our general method is summarized in Algorithm~\ref{alg:general}. It should be noted that in its most general case, error control is only provided if the multiple testing procedure works under arbitrary dependence. This is captured in the following theorem, whose proof follows immediately from the explanations above.

\begin{theorem}\label{theo:FDR_arbitrary_dep}
If $\bd$ is a monotone multiple testing procedure with FWER (or FDR) control under arbitrary dependence of the p-values, then Algorithm~\ref{alg:general} controls the FWER (or FDR)  at level $\alpha$.
\end{theorem}

For FWER control, the class of monotone multiple testing procedures under arbitrary dependence includes all Bonferroni-based procedures such as the Bonferroni-Holm \citep{holm1979simple}, the sequentially rejective graphical multiple testing procedure \citep{bretz2009graphical} and the fixed sequence method \citep{MHL}. For FDR control, the Benjamini-Yekuteli  procedure \citep{benjamini2001control} is  most common. In order to control the desired error rate under arbitrary dependence, the procedures usually need to protect against the worst case distribution. Therefore, improvements can be derived under additional assumptions. A typical assumption is positive regression dependence on a subset (PRDS), under which Hochbergs's \citep{hochberg1988sharper} and Hommel's \citep{hommel1988stagewise} procedure provide uniform improvements of Holm's procedure and the famous Benjamini-Hochberg (BH) procedure \citep{benjamini1995controlling} uniformly improves the Benjamini-Yekuteli procedure. However, if PRDS is required, caution with respect to the choice of anytime-valid permutation p-value and stopping time needs to be taken. We address this in the next section by the example of the BH procedure with anytime-valid BC p-values.



\ifarxiv
\begin{algorithm}
\caption{General multiple testing with anytime-valid permutation p-values}\label{alg:general}
 \textbf{Input:} Overall significance level $\alpha$, anytime-valid permutation p-values $(\pval_t^i)_{t\in \mathbb{N}}$, monotone p-value based multiple testing procedure $\bd=(d_i)_{i\in \{1,\ldots,M\}}$ and data-dependent stopping rules $\mathcal{S}^i$.\\ 
 \textbf{Output:} Stopping times $\tau_1,\ldots, \tau_M$ and rejection set $R$.
\begin{algorithmic}[1]
\State $A=\{1,\ldots,M\}$
\State $R=\emptyset$
\For{$t=1,2,\ldots$}
\For{$i\in A$} 
\State Check whether $H_0^i$ can be rejected by $\bd$.
\If{$d_i(P_{\min(t,\tau_1)}^1,\ldots,P_{\min(t,\tau_M)}^M)=1$}
\State $R=R\cup \{i\}$
\State $A=A\setminus \{i\}$
\State $\tau_i=t$
\EndIf
\If{$\mathcal{S}^i(\mathrm{data})=\mathrm{stop}$}
\State $A=A\setminus \{i\}$
\State $\tau_i=t$
\EndIf 
\EndFor
\If{$A=\emptyset$}
\State \Return $\tau_1,\ldots, \tau_M$, $R$
\EndIf
\EndFor
\end{algorithmic}
\end{algorithm}
\fi

\section{The Benjamini-Hochberg procedure with anytime-valid BC p-values\label{sec:bh}}

The Benjamini-Hochberg (BH) procedure \citep{benjamini1995controlling} rejects all hypotheses $H_0^i$ with p-value 
\begin{align}
\pval^i\leq m^* \alpha /M, \quad \text{where }m^*=\max\left\{ m\in \{1,\ldots,M\}: \sum_{i=1}^M \mathbbm{1}\{\pval^i\leq m\alpha /M\}\geq m \right\} \label{eq:threshold_BH}
\end{align}
with the convention $\max(\emptyset)=0$. In the following we write $m^*_t$ for the BH threshold at time $t\in \mathbb{N}$ when the p-values in \eqref{eq:threshold_BH} are replaced by anytime-valid p-values at time $t$. The BH procedure controls the FDR when the p-values are positive regression dependent on a subset (PRDS). In order to define PRDS we need the notion of an increasing set. A set $D\subseteq \mathbb{R}^M$ is increasing, if $z\in D$ implies $y\in D$ for all $y\geq z$.  

\begin{property}[PRDS \citep{benjamini2001control, finner2009false}]
    A random vector of p-values $\boldsymbol{\pval}$ is weakly PRDS on $I_0$ if for any null index $i \in I_0$ and increasing set $D\subseteq \mathbb{R}^M$, the function $x\mapsto \mathbb{P}(\boldsymbol{\pval}\in D 
\mid \pval^i \leq x)$ is nondecreasing in $x$ on $[0,1]$. We call $\boldsymbol{\pval}$ strongly PRDS on $I_0$, if $``\pval^i \leq x"$ is replaced with $``\pval^i = x"$. Furthermore, we denote $\boldsymbol{\pval}$ as independent on $I_0$, if $\pval^i$ is stochastically independent of the random vector $(\pval^1,\ldots,\pval^{i-1}, \pval^{i+1}, \ldots, \pval^{M})^T$ for all $i\in I_0$. \label{property:PRDS}
\end{property}

Strong PRDS implies weak PRDS, but for controlling the FDR with the BH procedure, weak PRDS on $I_0$ is sufficient \citep{finner2009false}. Note that the difference between these two PRDS notions is very minor and assuming strong PRDS instead of weak PRDS should not make a difference for most applications. However, differentiating between those two notions facilitates our proofs. When we just write PRDS in the remainder of this paper, we always mean weak PRDS. 
For example, the PRDS condition holds when $\boldsymbol{\pval}$ is independent on $I_0$. However, it also holds when there is some kind of positive dependence, which is an appropriate assumption in many trials. 


\subsection{FDR control with anytime-valid BC p-values under PRDS}
When PRDS is required, one needs to be careful with choosing the stopping time, as it could possibly induce some kind of negative dependence. For example, suppose we stop sampling for hypothesis $H_0^i$ as soon as it can be rejected by BH, meaning $\pval_{t}^i\leq m_t^{*} \alpha/M$. The smaller the other p-values, the larger $m_t^*$ and the sooner we could stop. This potentially induces some negative dependence between the p-values, even if the data used for calculating the different p-values is independent. 

In the following we prove that the BH procedure applied with the anytime-valid BC p-values controls the FDR under mild assumptions. First, we need the following lemma, showing that stopping early for rejection with the BH procedure cannot inflate the FDR. This is very important as it allows to adapt the number of permutations to the number of rejections.


\begin{lemma}\label{lemma:fast_stop}
  Let $\tau_i'$ be a stopping time for each $i\in \{1,\ldots,M\}$ such that the stopped anytime-valid p-values  $\pval_{\tau_1'}^1,\ldots,\pval_{\tau_M'}^M$ are PRDS and define \begin{align}\tau_i=\inf\{t\geq 1:\tau_i'=t \text{ or } \pval^i_t\leq m_t^* \alpha /M\}.\label{eq:stopping_for_rejection}\end{align} Then the BH procedure applied on $\pval_{\tau_1}^1,\ldots,\pval_{\tau_M}^M$ rejects the same hypotheses as if applied on $\pval_{\tau'_1}^1,\ldots,\pval_{\tau'_M}^M$ and therefore controls the FDR.  
\end{lemma}

Before we state our theorem on FDR control of BH with anytime-valid BC p-values, it should be noted that we cannot hope to prove FDR control under arbitrary assumptions, since we cannot expect our anytime-valid BC p-values to be PRDS if the data for the hypotheses do not exhibit some kind of non-negative dependence. To measure the dependency between the data, we use the \textit{limiting permutation p-values} $\plim^i=\lim_{t\to \infty} L_t^i/t$, $i\in \{1,\ldots,M\}$. The limiting permutation p-values can be interpreted as the ideal p-values we would use, if we could draw an infinite number of permutations. Furthermore, the limiting permutation p-values only depend on the data and not on the randomness induced by the generated test statistics. The following theorem proves FDR control of BH with the anytime-valid BC method, if we stop sampling for $H_0^i$ after $h$ losses are observed ($L_t^i=h$) or when $H_0^i$ can be rejected ($h/(t+h-L_t)\leq m_t^* \alpha /M, L_t<h$), whichever happens first. 

\begin{theorem}\label{theo:PRDS}
    Suppose the permuted samples are generated independently for all hypotheses and the vector of limiting permutation p-values $\boldsymbol{\pval}_{\mathrm{lim}}=(\plim^1,\ldots,\plim^M)$ is strongly PRDS on $I_0$. Then the anytime-valid BC p-values  $\boldsymbol{\pval}_{\boldsymbol{\gamma}(h)}^{\mathrm{avBC}}=(\pval_{\gamma_1(h)}^{\mathrm{avBC},1},\ldots,\pval_{\gamma_M(h)}^{\mathrm{avBC},M})$, where $\gamma_i(h)=\inf\{t\geq 1: L_t^i=h\}$,  are weakly PRDS on $I_0$. In particular, applying BH to $\boldsymbol{\pval}_{\boldsymbol{\tau}}^{\mathrm{avBC}}=(\pval_{\tau_1}^{\mathrm{avBC},1},\ldots,\pval_{\tau_M}^{\mathrm{avBC},M})$, where $\tau_i$ is defined as in \eqref{eq:stopping_for_rejection} for $\tau_i'=\gamma_i(h)$, controls the FDR.
\end{theorem}
Let us briefly make some practical notes on the assumptions of Theorem~\ref{theo:PRDS}. 

First, it is required that the limiting permutation p-values are \textit{strongly} PRDS, although BH generally works under weak PRDS. This might be an artifact of our proof technique. However, as mentioned before, the two assumptions do not make a big difference for practice. In particular, they are usually not proven but just assumed when there is no clear evidence against them. We see no specific use cases in which weak PRDS is a reasonable assumption but strong PRDS is not. 

Second, it is assumed that the permutations are sampled independently for the hypotheses. In practice, this is often not the case for computational reasons; e.g., by permuting the treatment indicators once, we can generate new data sets for all responses/hypotheses simultaneously. The dependence induced by such sampling mechanisms is difficult to track mathematically, which is why we focused on independent mechanisms. However, from a practical point of view, we see no reason why such a dependent mechanism should violate the PRDS condition. Also, the same consideration applies for the classical permutation p-values \eqref{eq:pperm} and we are not aware of any result that proves their PRDS under dependent sampling schemes.

In summary, we believe that Theorem~\ref{theo:PRDS} provides sufficient evidence for the validity of BH with the anytime-valid BC p-value, so that we can confidently replace the classical permutation p-value with the anytime-valid BC p-value in practice. In the next section, we show why this is a reasonable replacement, which is further supported by our simulations and real data analyses in Sections~\ref{sec:sim} and \ref{sec:real_data}.

\subsection{The synergy of anytime-valid BC p-values and BH\label{sec:magic}}

Concerns about using anytime-valid tests often relate to a loss of power and the possibility that the sample size could become very large in the event of unfavorable data constellations. In this section, we show that neither of these concerns apply to our anytime-valid BC p-value with the BH procedure.

At the beginning of this paper we noted that one must generate at least $B\geq M/\alpha-1$ permutations for each hypothesis to ensure that a rejection is possible when using BH with the permutation p-value \eqref{eq:pperm} regardless of the data of the other hypotheses. However, setting $B=\lceil M/\alpha \rceil-1$ is unsatisfactory as well. First, if the number of rejections is small, we would like to have more samples, as being unlucky with one generated test statistic could already reduce the power of the test. Hence, in practice one often sets $B=\lceil l M/\alpha \rceil-1$ for some small $l$ like $l=5$ or $l=10$. However, if the proportion of rejected hypotheses ends up being large, say $0.5$, we are effectively testing at a level of $\alpha/2$ with BH (see \eqref{eq:threshold_BH}). Drawing $\lceil l M/\alpha \rceil-1$ samples for a large $M$ seems like an excessive waste of computational resources in this case. For these reasons, the following adaptive sampling method would be desirable in order to adjust the number of samples to the BH threshold  $\alpha_{|R|}\coloneqq \alpha m^*/M= \alpha |R|/M\eqref{eq:threshold_BH}$, where $R$ is the rejection set of BH. 

Fix some constant $l$ and draw $B_M\coloneqq \lceil l/\alpha \rceil-1$ permutations for each hypothesis. If all hypotheses can be rejected at level $\alpha$, stop and make those rejections. If not, draw $\lceil lM/(M-1)\alpha\rceil -1-B_M$ further samples, meaning $B_{M-1}\coloneqq\lceil lM/(M-1)\alpha\rceil -1$ samples in total. If at least $(M-1)$ hypotheses can be rejected at level $\alpha(M-1)/M$, stop and make those rejections. If not, draw  $\lceil lM/(M-2)\alpha \rceil-1-B_{M-1}$ further samples and test at level $\alpha(M-2)/M$ and so on. In this way, we would draw $B=\lceil lM/[\max(|R|,1) \alpha] \rceil -1=\lceil l/\max(\alpha_{|R|}, \alpha/M) \rceil-1$ permutations for each hypothesis. Hence, this procedure adapts the number of permutations $B$ to the data-dependent level $\alpha_{|R|}$ of BH.

On first sight, such a proceeding might seem not permissible, since the number of samples of the permutation p-value \eqref{eq:pperm} cannot be adjusted to the data. However, it turns out that this is exactly what BH with anytime-valid BC p-values does (but, as we will see afterwards, in an even more efficient manner), and we have proven its FDR control in Theorem~\ref{theo:PRDS}.

\begin{theorem}\label{theo:magic}
    Let $\pval_t^{\mathrm{avBC},i}$, $i\in \{1,\ldots,M\}$, be the anytime-valid Besag-Clifford p-value \eqref{eq:avBC} for hypothesis $H_0^i$ with parameter $h$ at time $t$ and   $\pval_B^{\mathrm{perm},i}$ be the permutation p-value \eqref{eq:pperm} after $B$ permutations. Then, for any $m,i\in \{1,\ldots,M\}$,
    $$
    \pval_t^{\mathrm{avBC},i}\leq \alpha \frac{m}{M} \text{ for some }t \text{ with } L_t^i\leq h-1 \ \Leftrightarrow  \pval_{B_m}^{\mathrm{perm},i}\leq \alpha \frac{m}{M}, \text{ where } B_m=\left\lceil h \frac{M}{\max(m,1) \alpha} \right\rceil-1.
    $$
   Hence, BH applied on $\pval_{\tau_1}^{\mathrm{avBC},1},\ldots, \pval_{\tau_m}^{\mathrm{avBC},m}$, where $\tau_i$ is as in Theorem~\ref{theo:PRDS}, rejects the largest set $R\subseteq \{1,\ldots,M\}$ such that $\pval_{B_{|R|}}^{\mathrm{perm},i}\leq \alpha \frac{|R|}{M}$ for all $i\in R$. Furthermore, $\tau_i \leq B_{|R|}$ almost surely.
\end{theorem}



The above theorem particularly shows that we do not need to worry about a loss of power when using the anytime-valid BC p-value instead of the permutation p-value, as it can be interpreted as the permutation p-value applied with the data-adaptive and desirable number of permutations $B_{{|R|}}$. Thus, the anytime-valid BC method ensures that enough samples are drawn for each hypothesis, even if the number of rejections, and thus the BH threshold $\alpha_{|R|}$, is small.

Now, it might seem like the guarantee to draw enough samples must lead to an enormous amount of computational effort for certain data constellations (e.g, if the number of rejections is small) and that the actual realized number of permutations is unpredictable. In fact, Theorem~\ref{theo:magic} only implies an upper bound of $\lceil M h /\alpha \rceil -1$ for the number of permutations for each hypothesis, which can be huge if $M$ is large.  Hence, the users might be worried that they cannot perform the entire procedure in a reasonable amount of time for unfavorable data constellations. However, note that the maximum possible amount of samples $\lceil M h /\alpha \rceil -1$ can only be achieved for a hypothesis if all other hypotheses have already been stopped for futility. In the following proposition, we exploit this to provide a worst case bound for the average number of permutations per hypothesis. 

\begin{proposition}\label{prop:upper_bound}
    Consider applying BH on anytime-valid BC p-values with parameter $h$ and stopping times $\tau_i$, $i\in \{1,\ldots,M\}$, as in Theorem~\ref{theo:PRDS}. Let $\bar{\tau}=\frac{1}{M} \sum_{i=1}^M \tau_i$ be the average number of permutations per hypothesis. Then,
    \begin{align}
        \bar{\tau}\leq \left\lfloor \frac{h}{\alpha} -1 \right\rfloor +  \frac{h}{\alpha} \sum_{t=\left\lfloor \frac{h}{\alpha} \right\rfloor}^{\left\lfloor \frac{M h}{\alpha} -2 \right\rfloor} \frac{1}{(t+1)}  =\mathcal{O}(\log M) \quad \text{surely}. \label{eq:worst_bound}
    \end{align}
\end{proposition}

Consequently, even if an adversarial chooses the number of losses and when the losses for the hypotheses occur, the number of permutations per hypothesis increases at most logarithmically in the number of hypotheses $M$. This is much better than the lower bound $M/\alpha-1$ we derived for the permutation p-value. For example, if $h=10$, $\alpha=0.1$ and $M=1000$, the above upper bound tells us that we never need to draw more than $789$ samples per hypothesis. And if $M$ increases to one million, the average number of permutations per hypothesis is still smaller than $1500$. Also note that \eqref{eq:worst_bound} is a worst case bound, our simulations and real data experiments show that the average number of samples is usually even much smaller.

To summarize, Theorem~\ref{theo:magic} shows that the anytime-valid BC method makes the same rejections as the permutation p-value with a data-adaptive number of permutations, which ensures that enough samples are drawn to make powerful decisions with BH. Furthermore, Proposition~\ref{prop:upper_bound} implies that the number of samples drawn is always small and increases only slowly in the number of the hypotheses.

\section{A general approach: Anytime-valid permutation testing by betting\label{sec:anytime_betting}}

Up until this point, we have mainly focused on the anytime-valid BC p-value. However, there are infinitely many other anytime-valid permutation p-values one could construct and our general Algorithm~\ref{alg:general} as well as Theorem~\ref{theo:FDR_arbitrary_dep} work for arbitrary anytime-valid p-values. For this reason, we now recap the general construction of anytime-valid permutation tests proposed by  \citet{fischer2024sequential} and show how it can be made suitable for multiple testing. However, it should be noted that due to its simplicity and desirable behavior with the BH procedure (see Section~\ref{sec:bh}), the anytime-valid BC p-value  is our recommended approach for most situations (see Section~\ref{sec:conclusion} for more details). Hence, practically-oriented readers may skip this section.

 \citet{fischer2024sequential} introduced a general approach to anytime-valid permutation testing based on a testing by betting technique \citep{grunwald2020safe, shafer2021testing}. In particular, the anytime-valid BC method is a concrete instance of their algorithm. However, it also allows to construct many other anytime-valid permutation tests. We start with recapping their method for a single hypothesis. At each step $t=1,2,\ldots$, the statistician bets on the outcome of the loss indicator $I_t= \mathbbm{1}\{Y_t\geq Y_0\}$ by specifying a betting function $B_t:\{0,1\}\rightarrow \mathbb{R}_{\geq 0}$. If $I_t=r$, the wealth of the statistician gets multiplied by $B_t(r)$, $r\in \{0,1\}$. Starting with an initial wealth of $W_0=1$, after $t$ rounds of gambling the wealth of the statistician is given by
$$
W_t=\prod_{s=1}^t B_s(I_s),
$$
where the concrete structure of the wealth depends on the betting strategy used. 
The wealth process $(W_t)_{t\in \mathbb{N}}$ is a test martingale (nonnegative martingale with $W_0=1$ under $H_0$) with respect to the filtration generated by the indicators $(\mathcal{I}_t)_{t\in \mathbb{N}}$, if 
$$
B_t(0)\frac{t-L_{t-1}}{t+1}+B_t(1)\frac{1+L_{t-1}}{t+1}=1.
$$
By the optional stopping theorem, this implies that $(W_t)_{t\in \mathbb{N}}$ is an anytime-valid e-value, meaning $\mathbb{E}_{H_0}[W_\tau]\leq 1$ for any stopping time $\tau$. In addition, by Ville's inequality, 
\begin{align}
    \mathbb{P}_{H_0}(\exists t\in \mathbb{N}: W_t\geq 1/\alpha)\leq \alpha. \label{eq:Ville_ineq}
\end{align}
 This shows that a hypothesis could be rejected at level $\alpha$ as soon as the wealth is larger or equal than $1/\alpha$. Ville's inequality is often exploited \citep{ramdas2023game, ramdas2024hypothesis} by defining an anytime-valid p-value as
 \begin{align}
\pval_{t}:=\frac{1}{\max_{s=1,\ldots,t}W_s}. \label{eq:anytime_p-value}
\end{align}
 However, we introduce a more powerful approach in Section~\ref{sec:alpha_calibration}, which is particularly useful for multiple testing. But first, we recap a particular instance of the above testing by betting algorithm.


\subsection{The binomial mixture strategy\label{sec:binomial_mixture}}
As a special instance of their general approach, \citet{fischer2024sequential} proposed the \textit{binomial mixture strategy}, which particularly leads to desirable asymptotic behavior. For a predefined parameter $c\in [0,1]$, the wealth of the binomial mixture strategy is given by 
\begin{align}
\Wbm_t^c=(1-\mathrm{Bin}(L_t;t+1,c))/c, \label{eq:bm}
\end{align}
where $\mathrm{Bin}(L_t;t+1,c)$ is the CDF of a binomial distribution with size parameter $t+1$ and probability $c$. \citet{fischer2024sequential} showed that 
\begin{align}\Wbm_t^c|\{\plim=\ptrue\} \stackrel{a.s.}{\to}\begin{cases}
        1/c, &\text{ if } \ptrue \in [0,c) \\
        0, &\text{ if } \ptrue \in (c,1]
    \end{cases} \label{eq:asymp_bm}
    \end{align}   for $t\to \infty$,  where $\plim$ is the limiting permutation p-value. Therefore, if we choose $c<\alpha$ for some predefined significance level $\alpha$, the binomial mixture strategy rejects almost surely after a finite number of permutations if $\plim <c$. Hence, we can make the loss compared to the limiting permutation p-value arbitrarily small by choosing $c<\alpha$ arbitrarily close to $\alpha$.


   In multiple testing scenarios the level of an individual hypothesis might not be predefined and depends on the rejection or non-rejection of the other hypotheses. In the following section, we will show how we can still use the binomial mixture strategy and all other $\alpha$-dependent betting strategies in multiple testing procedures.  

\subsection{$\alpha$-dependent p-value calibration\label{sec:alpha_calibration}}

Many multiple testing methods take p-values $\pval^1,\ldots,\pval^M$ as inputs and reject $H_0^i$, if $\pval^i\leq \alpha^i$ for some potentially data-dependent individual significance level  $\alpha^i\in [0,1)$. Therefore, we do not know in advance at which individual significance level the hypothesis $H_0^i$ will be tested at. However, most betting strategies derived by \citet{fischer2024sequential}, including the binomial mixture strategy defined in Section~\ref{sec:binomial_mixture}, require the significance level as input and thus it need to be fixed in advance. In this subsection, we show how this can be circumvented by calculating an anytime-valid e-value for each level $\alpha\in [0,1]$ and then calibrating these into an anytime-valid p-value by taking the smallest level $\alpha$ at which the corresponding e-value can reject the null hypothesis.

For each null hypothesis $H_0^i$ we choose a betting strategy as described in Section~\ref{sec:anytime_betting} we would use if we test at level $\alpha' \in (0,1)$. Let $W_{t}^{i,\alpha'}$ be the wealth of the strategy for hypothesis $H_i$ and level $\alpha'$ after $t$ permutations. Going forward, we assume that for all $\alpha_1' < \alpha_2'$, 
\begin{align} \{W_{t}^{i,\alpha_1'}\geq 1/\alpha_1'\} \subseteq \{W_{t}^{i,\alpha_2'}\geq 1/\alpha_2'\}.\label{eq:mon_alpha}\end{align}
This means that if our strategy for level $\alpha_1'$ rejects at level $\alpha_1'$, our strategy for some larger level $\alpha_2'>\alpha_1'$, needs to reject at level $\alpha_2'$ as well. In particular, this is satisfied if we use for each $\alpha'\in (0,1)$ the binomial mixture strategy with parameter $c=b\alpha'$ for some constant $b\in (0,1)$. To see this, note that 
$$\Wbm_t^c \geq 1/\alpha' \Longleftrightarrow \frac{1-\mathrm{Bin}(L_t;t+1,b \alpha')}{b} \geq 1  $$
and $(1-\mathrm{Bin}(\ell;t+1,b \alpha'))/b$ is monotone increasing in $\alpha'$. In the following, we therefore also use the parameter $b$ instead of $c$ to parameterize the binomial mixture strategy in an $\alpha$-independent way.

We define the $p$-value for the $i$-th hypothesis at step $t\in \mathbb{N}$ as \begin{align}\pval_t^i=\inf\{\alpha\in (0,1)|\exists s\in \{1,\ldots,t\}: W_{s}^{i,\alpha}\geq 1/\alpha\}.\label{eq:cont_p-value}\end{align}
For example, the anytime-valid version of the BC method \eqref{eq:avBC} can be obtained using the betting approach with such an $\alpha$-dependent calibration, however, \citet{fischer2024sequential} just used this implicitly and did not write this approach down in its general form. 
 Note that if we use the same strategy for each $\alpha$, then \eqref{eq:mon_alpha} is trivially satisfied and the p-value in \eqref{eq:cont_p-value} reduces to the one in \eqref{eq:anytime_p-value}. 

\begin{proposition}
    If \eqref{eq:mon_alpha} holds, then $(\pval_t^i)_{t\in \mathbb{N}}$ in~\eqref{eq:cont_p-value} is an anytime-valid p-value.
\end{proposition}
The proof follows by combining Ville's inequality \eqref{eq:Ville_ineq} and Condition~\eqref{eq:mon_alpha} to note that for all $\tilde{\alpha}>\alpha$ and every stopping time $\tau_i$ we have 
$$
\mathbb{P}_{H_0^i}(\pval_{\tau_i}^i\leq \alpha) \leq \mathbb{P}_{H_0^i}(\exists t\in \{1,\ldots,\tau_i\}: W_{t}^{i,\tilde{\alpha}}\geq 1/\tilde{\alpha}\}) \leq \tilde{\alpha}.
$$
Letting $\tilde{\alpha}\to \alpha$, proves the assertion.
It might seem computationally challenging to calculate the p-value \eqref{eq:cont_p-value} at each step $t$. However, 
we do not need to evaluate $W_t^{i,\alpha}$ for each $\alpha\in [0,1]$, but in order to compare $\pval_t^i$ with some individual significance level $\alpha^i$ we can just check whether $W_t^{i,\alpha^i}\geq 1/\alpha^i$, saving a lot of computational effort. Sometimes the p-value \eqref{eq:cont_p-value} also has a closed form, as it is the case for the anytime-valid BC method \eqref{eq:avBC}  \citep{fischer2024sequential}.

\subsection{FDR control with BH and general anytime-valid p-values}

In this section, we discuss FDR control when using BH with general anytime-valid p-values. First, note that Lemma~\ref{lemma:fast_stop} does not make any assumptions on the anytime-valid p-values. Hence, stopping for rejection with general anytime-valid p-values in the BH procedure does not inflate the FDR  and it is sufficient to prove PRDS of the anytime-valid p-values at the futility stops.

Thus, it immediately follows that FDR control is guaranteed if the data for each hypothesis is independent and the stop for futility depends only on the data of the corresponding hypothesis. To make it more concrete, one could write such a stop for futility quite generally as \begin{align}\tau_i'=\inf\{t\geq 1: L_t^i>z_t^i\}, \label{eq:stop_futility}\end{align} where the parameters $z_t^i$ are constants. For example, the stopping time of the Besag-Clifford method $\gamma_i(h,B)$ is obtained by $z_1^i=\ldots=z_{B-1}^i=h-1$ and $z_B=-1$.  

\begin{proposition}\label{prop:futility_ind}
    Suppose the permuted samples are generated independently for all hypotheses and the vector of limiting p-values $\boldsymbol{\pval}_{\mathrm{lim}}$ is independent on $I_0$ (see Property~\ref{property:PRDS}). Furthermore, for each $i\in \{1,\ldots,M\}$, let the stopping time  $\tau_i'$ be given by \eqref{eq:stop_futility}. Then the stopped anytime-valid p-values $\boldsymbol{\pval_{\tau'}}=(\pval_{\tau_1'}^1,\ldots,\pval_{\tau_M'}^M)$ are independent on $I_0$. In particular, applying BH to $\boldsymbol{\pval_{\tau}}=(\pval_{\tau_1}^1,\ldots,\pval_{\tau_M}^M)$, where $\tau_i$ is defined as in \eqref{eq:stopping_for_rejection}, controls the FDR.
\end{proposition}


If there is some positive dependence between the data for the different hypotheses, proving PRDS for the anytime-valid p-values becomes more difficult. The reason for this, and difference to the anytime-valid BC p-value case (see Theorem~\ref{theo:PRDS}), is that $\pval_t^i$ is not only a function of the losses at step $t$ but might depend on the losses of all previous steps $L_1^i,\ldots, L_t^i$. This can lead to situations in which a larger p-value $\pval_t^i$ yields more evidence against a large $\plim^i$, which makes it hard to use that $\boldsymbol{\pval}_{\mathrm{lim}}$ is PRDS for proving that $\boldsymbol{\pval}_t$ is PRDS. 
\begin{example}
    Consider a simple example in which for some $x<x^*$: $\pval_3^i\leq x$ is equivalent to $L_1^i\leq 0 \cup L_2^i\leq 0 \cup L_3^i\leq 0$, which reduces to $L_1^i\leq 0$, and $\pval_3^i\leq x^*$ is equivalent to $L_1^i\leq 0 \cup L_2^i\leq 0 \cup L_3^i\leq 1$, which reduces to $L_1^i\leq 0 \cup L_3^i\leq 1$.  Recall that $L_t^i$ follows a mixture binomial distribution with size parameter $t$ and random probability $\plim^i$. With this, under the assumption that $\plim^i$ follows a uniform distribution (which is true under $H_0^i$), one can check that $\mathbb{P}(\plim^i> 0.9\mid L_1^i\leq 0)=0.01>0.0091=\mathbb{P}(\plim^i> 0.9\mid L_1^i\leq 0 \cup L_3^i\leq 1)$. To understand why this is the case, note that $L_1^i\leq 0$ yields the following set of possible loss sequences $\{(0,1,1),(0,0,1),(0,1,0),(0,0,0)\}$, while $L_1^i\leq 0 \cup L_3^i\leq 1$ only adds $(0,0,1)$ to this set. For this reason, $L_1^i\leq 0 \cup L_3^i\leq 1$ provides more evidence against small values for $\plim^i$, but also more evidence against very large values of $\plim^i$. Therefore, $\plim^i$ is not PRDS on $\pval_t^i$ in general. 
\end{example}

We expect the effect from the above phenomenon, that larger bounds for the number of losses can lead to a lower probability for large values of the limiting permutation p-value, to be very minor and not inflating the FDR. If we choose a betting strategy with nonincreasing wealth in the number of losses (like the binomial mixture strategy \eqref{eq:bm}) and use the stop for futility \eqref{eq:stop_futility}, then the stopped anytime-valid p-values $\pval_{\tau_i'}^i$, $i\in \{1,\ldots,M\}$, are coordinatewise nondecreasing and nonrandom functions of the number of losses $(L_t^i)_{t\in \mathbb{N}}$. Due to De Finetti's theorem, $\boldsymbol{\pval}_{\mathrm{lim}}$ being strongly PRDS implies that $(L_{t_1}^1,\ldots, L_{t_M}^M)$ is PRDS for all $t_1,\ldots,t_M \in \mathbb{N}$ (see also the proof of Theorem~\ref{theo:PRDS}). Hence, in this case $\pval_{\tau_1'}^1,\ldots, \pval_{\tau_M'}^M$ potentially have some kind of positive association as well and therefore we argue that it is reasonable to replace the classical permutation p-values with our anytime-valid ones in the BH procedure. Even if the $z_t^i$ in \eqref{eq:stop_futility} depend in a nonincreasing way on the losses of the other hypotheses, we do not see any reason why this should hurt the PRDS of the stopped p-values. This is supported by our simulations in Section~\ref{sec:sim_FDR}. However, a general proof of FDR control under PRDS of the limiting permutation p-values seems difficult or even impossible.




\section{Simulated experiments\label{sec:sim}}
In this section, we aim to characterize the behavior of applying the BH procedure with the proposed sequential permutation p-values using simulated data. We consider an independent Gaussian mean multiple testing problem. That means, each hypothesis is given by $H_0^i:\mathbb{E}[Y_0^i]=0$, $i\in \{1,\ldots,M\}$, where $Y_0^i$ follows a standard normal distribution under the null hypothesis and a shifted standard normal distribution with mean $\mu_A$ under the alternative. The probability for a hypothesis being false was set to $\pi_A\in [0,1]$ and the generated test statistics $Y_1^i,\ldots, Y_B^i$ were sampled from $N(0,1)$. Note that this matches the setting described in Section~\ref{sec:old_paper}. The R code for reproducing the simulations is available at \url{github.com/fischer23/MC-testing-by-betting/}.

\subsection{Power and average number of permutations\label{sec:sim_power_nperm}}
First, we evaluate the power and average number of permutations per hypothesis using the aggressive strategy, the BC method \citep{besag1991sequential} with $h=10$, the anytime-valid generalization of the BC method with $h=10$ \eqref{eq:avBC}, the binomial mixture strategy \eqref{eq:bm} with $b=0.9$, the AMT algorithm by \citet{zhang2019adaptive} with $\delta=0.01$ and the classic permutation p-value \eqref{eq:pperm} with $B=\numprint{10000}$. The anytime-valid BC method was applied as in Theorem~\ref{theo:PRDS} and the binomial mixture strategy as described in \ifarxiv Appendix~\ref{sec:detail_bm}\fi, but all methods were stopped after a maximum number of $B=\numprint{10000}$ permutations. In addition to applying the classical permutation p-value with $B=\numprint{10000}$, we also evaluated it for $B=200$ as a reference. It should be noted that the AMT algorithm was applied at an overall level of $\alpha-\delta$ such that FDR control at level $\alpha$ is provided \citep{zhang2019adaptive}.

\begin{figure}[h!]
\centering
\includegraphics[width=0.95\textwidth]{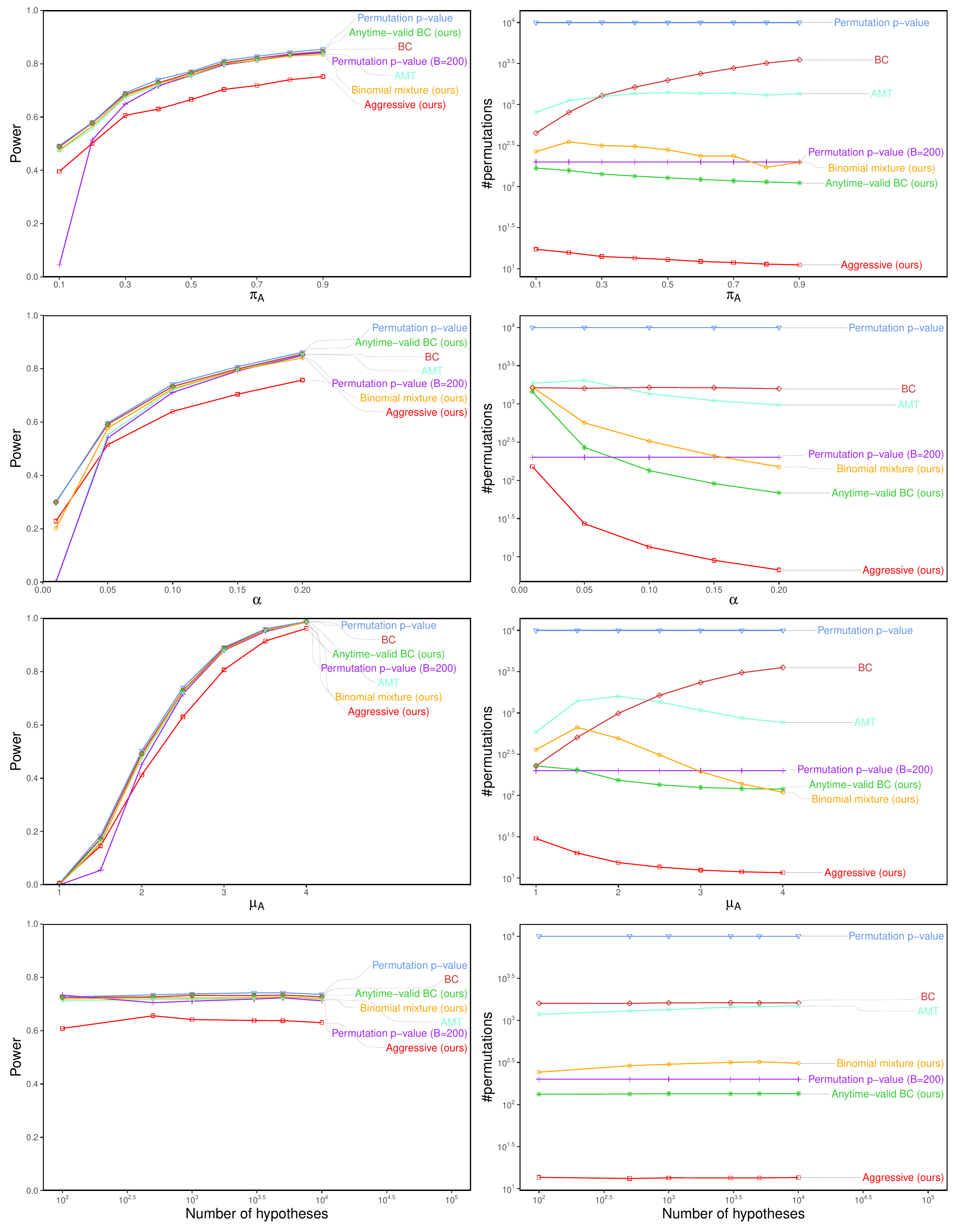}
\caption{Power and average number of permutations per hypothesis for varying simulation parameters obtained by applying the BH procedure to different (sequential) permutation testing strategies.  The anytime-valid BC method and the binomial mixture strategy lead to marginally less power than the classical permutation p-value, while reducing the number of permutations by orders of magnitude. \label{fig:sim_results} }\end{figure}

The results in Figure~\ref{fig:sim_results} were obtained by averaging over $10$ independently simulated trials. Similar to \citet{zhang2019adaptive}, we set the standard values of the simulation parameters to $\pi_A=0.4$, $\mu_A=2.5$, $\alpha=0.1$  and $M=1000$, while one of them was varied in each of the plots. 

It can be seen that the anytime-valid BC method and binomial mixture strategy lead to a similar power as the classical permutation p-value for $B=\numprint{10000}$ permutations, while being able to reduce the number of permutations by orders of magnitude. When the number of permutations of the classical permutation p-value is reduced to $200$, the power reduces substantially, particularly if the proportion of false hypotheses, strength of the alternative  or significance level is small. Since the anytime-valid BC and binomial mixture method also need approximately $200$ permutations on average, this shows that the performance of these sequential methods cannot be accomplished by the permutation p-value with a fixed number of permutations. The anytime-valid BC method always makes the same rejections as the classical BC method and thus has the same power, but reduces the number of permutations significantly.


 The AMT algorithm was outperformed by the anytime-valid BC method and the binomial mixture strategy in terms of power and number of permutations in all considered scenarios. The use of the aggressive strategy can be reasonable when the main goal is to reduce the number of permutations, while a power loss is acceptable. Overall, the anytime-valid BC method performed best in terms of power and number of permutations.

It should be noted that the behavior of the methods does not change much with the number of hypotheses, since the other parameters remain fixed, which implies that $m_t^*/M$, and thus the significance level of BH procedure \eqref{eq:threshold_BH}, remain approximately constant. Lastly, we would like to highlight that the results for all these different constellations of the simulation parameters were obtained with the same hyperparameters for the sequential methods, which thus seem to be universally applicable choices.

\subsection{Early reporting of decisions}
In the previous subsection, we have shown that a lot of permutations can be saved by our sequential strategies, while the power remains similar to the classical permutation p-value. However, reducing the total number of permutations is not the only way of increasing efficiency with sequential permutation tests.
We can also report already made decisions, and particularly rejections, before the entire process has stopped. This might not be important if the entire procedure only takes some hours or one day to run. However, in large-scale trials generating all required permutations and performing the tests can take up to several months or even longer. This can slow down the research process, as the results are crucial for writing a scientific paper or identifying follow up work. Therefore, it would be helpful to be able to already report the unambiguous decisions at an earlier stage of the process. Note that this is not possible with the AMT algorithm by \citet{zhang2019adaptive}, since no decisions can be obtained until the entire process has finished. 

In Figure~\ref{fig:stops_for_rejection} we show the distribution of the rejection times in an arbitrary simulation run using the anytime-valid BC and the binomial mixture strategy in the standard setting described in Section~\ref{sec:sim_power_nperm}. In this case, the binomial mixture strategy rejected $302$ and the anytime-valid BC method $305$ hypotheses, while the former needed a mean number of $462$ and the latter of $327$ permutations to obtain a rejection. However, the distribution of the stopping times looks very different. More than $50\%$ of the rejections made by the binomial mixture strategy were obtained after $147$ permutations and more than $75\%$ after less than $185$ permutations, while the rest is distributed quite broadly up to $\numprint{10000}$ permutations. In contrast, all rejections by the anytime-valid BC procedure were made at time $327$ such that the entire process was stopped at that time. 

 This can be explained by Theorem~\ref{theo:magic}. The rejections made by the anytime-valid BC p-values are the same as with the permutation p-values \eqref{eq:pperm} for a particular data-dependent number of permutations.
 Hence, we often sample until this data-dependent time and then reject all remaining hypotheses such that in many cases (but not always) all rejections are obtained at the same time. 


Consequently, while the anytime-valid BC method needed less permutations in total, the binomial mixture strategy would be more useful if early reporting of the rejections is desired, since the majority of decisions is obtained much faster. In practice, the binomial mixture strategy could be used to start interpreting the results while keeping to generate further permutations in order to increase the total number of discoveries. Indeed, if the data for some of the undecided hypotheses looks promising, sampling could be continued after the first $\numprint{10000}$ permutations to possibly make even more rejections. In this case, one could also consider increasing the precision of the binomial mixture strategy by choosing the parameter $b$ closer to $1$, since a larger average number of permutations might be acceptable when the majority of decisions is obtained fast.

\begin{figure}[h!]
\centering
\includegraphics[width=11cm]{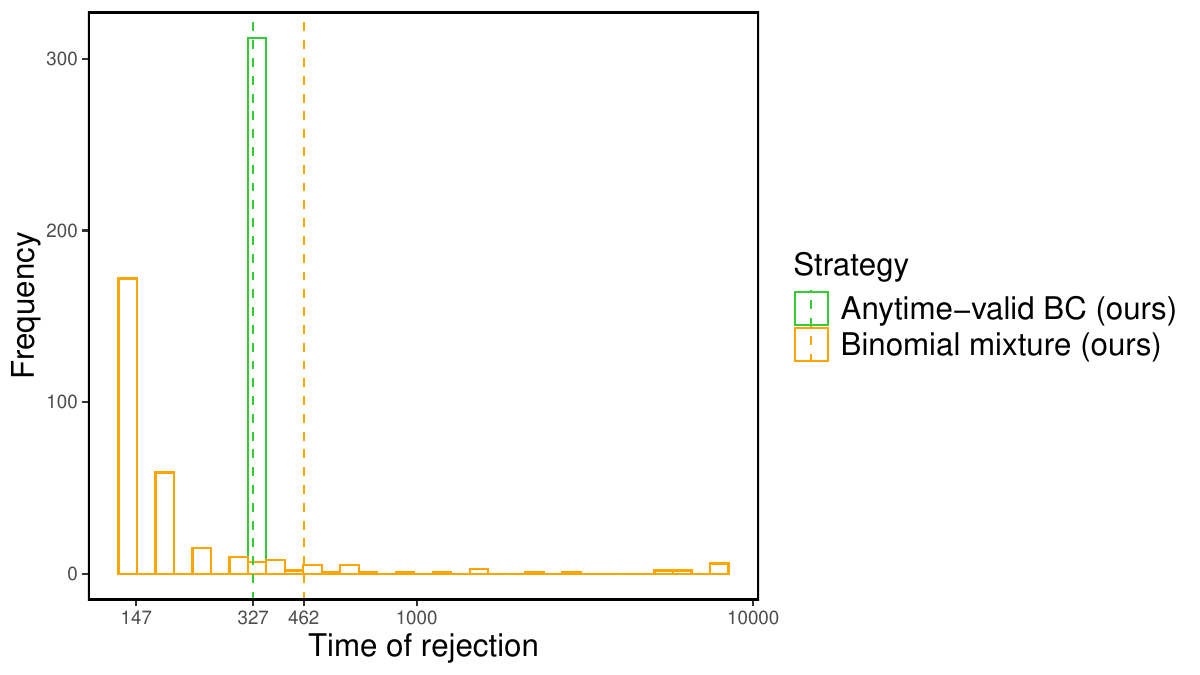}
\caption{Distribution of the rejection times obtained by the BH procedure applied with the anytime-valid BC and the binomial mixture strategy in a standard simulation run according to Section~\ref{sec:sim_power_nperm}.  All discoveries by the anytime-valid BC method were made at time $328$, while the binomial mixture strategy made more than $75\%$ of the rejections after less than $180$ permutations but needed more time on average (dashed vertical line). \label{fig:stops_for_rejection} }\end{figure}

\subsection{FDR control\label{sec:sim_FDR}}
To evaluate the FDR control of BH with our anytime-valid p-values, we generate the vector $(Y_0^1, \ldots, Y_0^M)$ as multivariate normal with mean $0$ and $\mathrm{cov}(Y_0^i,Y_0^j)=\rho$ for $i\neq j$, $\rho \in \{0,0.1,0.3,0.5,0.7,0.9\}$. In this case, the limiting permutation p-values are strongly PRDS \citep{benjamini2001control}. The other simulation parameters and parameters for the permutation testing strategies were chosen as in our standard setting introduced in Section~\ref{sec:sim_power_nperm}. The FDR is obtained by averaging over $1000$ independent trials. All procedures controlled the FDR at the prespecified level $\alpha=0.01$ for all considered parameters $\rho$ (see Figure~\ref{fig:sim_results_PRDS}). This is not surprising, since the BH procedure with the aggressive strategy, the anytime-valid BC method and the classical permutation p-value control it provably and we gave reasonable arguments why it should with the binomial mixture strategy as well. Also note that in this case the BH procedure controls the FDR at level $0.06=0.6\alpha$, since the proportion of false hypotheses was set to $\pi_A=0.4$ \citep{benjamini1995controlling}.

\begin{figure}[h!]
\centering
\includegraphics[width=11cm]{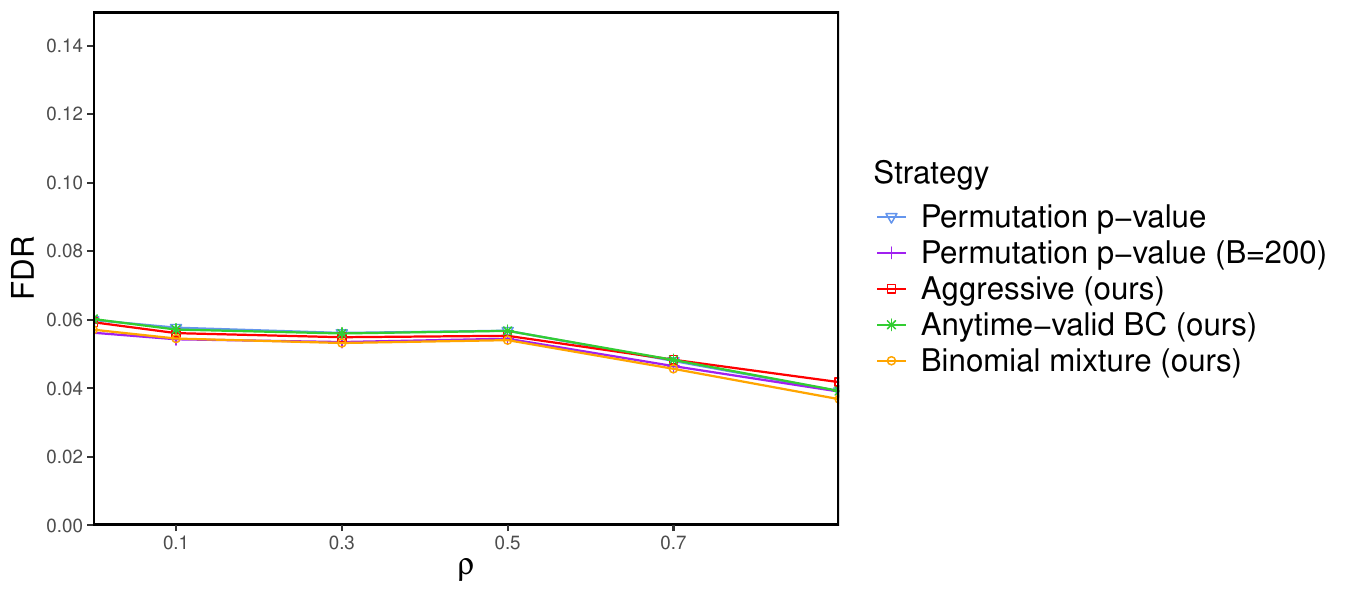}
\caption{FDR for increasing data correlation obtained by applying the BH procedure to different (sequential) permutation testing strategies.  The empirical FDR is below the level $\alpha \pi_0$ for all procedures, where $\pi_0$ is the proportion of true hypotheses\label{fig:sim_results_PRDS}  }\end{figure}

\section{Real data analysis\label{sec:real_data}}

We apply our sequential permutation testing method to analyze an RNA sequencing (RNA-seq) dataset, carefully profiling the running time of our method and comparing it to competitors. 

Typically, RNA-seq data are analyzed using a parametric regression method, such as DESeq2 \cite{Love2014} or Limma \cite{Ritchie2015}. However, in a recent, careful analysis of 13 RNA-seq datasets, \citet{Li2022} found that popular tools for RNA-seq analysis --- including DESeq2 --- did not control type-I error on negative control data (i.e., data devoid of signal). This likely was because the parametric assumptions of these methods broke down. Li~et~al.\ instead recommended use of the Mann-Whitney test \cite{Mann1947} --- a classical, nonparametric, two-sample test --- for RNA-seq data analysis.\footnote{According to \citet{Li2022}, the RNA-seq dataset should contain at least 16 samples for the Mann-Whitney test to have sufficient power.}
	
Li~et~al.\ applied the Mann-Whitney test to analyze several datasets, including a dataset of human adipose (i.e., fat) tissue generated by the Genotype-Tissue Expression (GTEx) project \cite{Lonsdale2013}. Subcutaneous (i.e., directly under-the-skin) adipose tissue and visceral (i.e., deep within-the-body) adipose tissue samples were collected from 581 and 469 subjects, respectively. RNA sequencing was performed on each of these $n = 1,050$ tissue samples to measure the expression level of each of $M = 54,591$ genes. This experimental procedure yielded a tissue-by-gene expression matrix $Z \in \mathbb{N}^{n \times M}$ and a binary vector $W \in \{0,1\}^{n}$ indicating whether a given tissue was subcutaneous ($W_i = 0$) or visceral ($W_i = 1$). Li et al.\ sought to determine which genes were differentially expressed across the subcutaneous and visceral tissue samples. To this end, for each gene $j \in \{1, \dots, M\}$, Li et al.\ performed a Mann-Whitney test to test for association between the gene expressions $Z_{1,j}, \dots, Z_{n,j}$ and the indicators $W_1, \dots, W_n$, yielding p-values $P^1, \dots, P^M$. Finally, Li~et~al.\ subjected the p-values to a BH correction to produce a discovery set. Li~et~al.\ used the asymptotic version of the Mann-Whitney test, as it is more computationally efficient than its finite-sample analogue in high-multiplicity settings. The finite sample Mann-Whitney test generally is calibrated via permutations.

We sought to explore whether our sequential permutation testing procedure would enable computationally efficient application of the finite-sample Mann-Whitney test to the RNA-seq data. We conducted an experiment in which we applied three methods to analyze the data: (i) the asymptotic Mann-Whitney p-value (implemented via the \texttt{wilcox.test()} function in R); (ii) a permutation p-value based on the Mann-Whitney statistic using a fixed number $B = 5M/\alpha$ of permutations across hypotheses; and (iii) the anytime-valid BC p-value based on the Mann-Whitney test statistic with Algorithm~\ref{alg:general}. We implemented our anytime-valid BC method in C++ for maximum speed (see \ifarxiv Appendix~\ref{appn:rna_seq_info} \fi  for more details). We set the tuning parameter $h$ to $15$ in the latter method, and we set the nominal FDR across methods to $\alpha = 0.1$. The classical permutation p-value was slowest, taking over three days and 14 hours to complete. The asymptotic p-value and the anytime-valid permutation p-value, on the other hand, were much faster, running in about two minutes and three and a half minutes, respectively. All three methods rejected a similar proportion of the null hypotheses ($56.4 - 56.9$\% across methods). 

While the asymptotic MW test and finite-sample MW test rejected a similar number of hypotheses, this is not always the case. The asymptotic approximation can break down, leading to a loss of power or an inflated number of false discoveries. In a separate study, \citet{barry2025permuted} evaluated the asymptotic MW test and the finite-sample MW test calibrated via the anytime-valid BC procedure (as proposed in this work) on a bulk RNA-seq dataset comprising 16 samples. The asymptotic MW test made zero discoveries, while the finite-sample MW test made 292 discoveries, including an internal positive control known to be differentially expressed (\citet{barry2025permuted}, Table 6c). Thus, the anytime-valid permutation p-value inherits the strengths of both the asymptotic p-value and the classical permutation p-value. Like the asymptotic p-value, the anytime-valid permutation p-value is fast to compute, and like the classical permutation p-value, the anytime-valid permutation p-value avoids asymptotic assumptions.


\begin{table}\caption{Comparison of the BH procedure applied with the asymptotic Mann-Whitney p-value, the classical permutation p-value based on the Mann-Whitney statistic, and Algorithm~\ref{alg:general} with the anytime-valid BC p-value based on the Mann-Whitney statistic. All three methods were applied to analyze the adipose RNA-seq data. Our proposed anytime-valid permutation p-value was nearly as fast as the asymptotic p-value while avoiding asymptotic assumptions.}

\centering
\begin{tabular}{|c|c|c|}
	\hline
	\textbf{Method} & \textbf{Proportion rejected ($\%$)} & \textbf{Running time} \\
	\hline
	Asymptotic p-value & $56.6\%$  & 2m 4s \\
	\hline
	Classical permutation p-value & $56.4\%$  & 3d 14h 24m  \\
	\hline
	Anytime-valid BC p-value (ours) & $56.9\%$ & 3m 27s \\
	\hline
\end{tabular}
\end{table}

An R/C++ package providing fast implementations of several of the methods introduced in this paper is available at \url{github.com/timothy-barry/adaptiveperm/} and the R code to reproduce this real data analysis is available at \url{github.com/timothy-barry/sequential\_perm\_testing}.

\section{Conclusion\label{sec:conclusion}}
There are several advantages of using the proposed anytime-valid permutation p-values in large-scale multiple testing problems:
\begin{enumerate}
    \item The number of permutations is automatically adapted to the data-dependent significance level of the multiple testing procedure (see Theorem~\ref{theo:magic} for the anytime-valid BC method). In this way, it is ensured that enough samples are drawn to make powerful decisions.
    \item The required number of permutations is reduced by orders of magnitude. In particular, we proved in Proposition~\ref{prop:upper_bound} that the number of permutations required per hypothesis in the worst case scales only logarithmically with the number of hypotheses when using BH with anytime-valid BC p-values (compared to linearly with the permutation p-value). Furthermore,  simulations showed that there is no fixed number of permutations such that the classical permutation p-value performs similar to the sequential permutation p-values. The anytime-valid BC method performed best in terms of power and required number of permutations.
    \item Due to the data-adaptive nature of our methods, there is no need to adapt the parameter choice to the unknown data generating process. We found $h=10$ for the anytime-valid BC method and $b=0.9$ for the binomial mixture strategy to deliver good results in a wide variety of settings. Depending on whether the main focus lies in increasing precision or reducing the number of permutations, $h$ and $b$ could also be chosen larger or smaller, respectively.
    \item Decisions can be reported early such that one can begin follow-up work before the entire testing process is finished. In particular, the binomial mixture strategy was found to be useful for such a proceeding.
\end{enumerate}

 In this paper, we showed how anytime-valid permutation tests can be used with p-value based multiple testing procedures. Another type of multiple testing procedures does not rely on individual p-values and explicitly uses  permutation tests to adapt to the unknown dependence structure of the data. This includes the famous MaxT approach for FWER control by Westfall \& Young \citep{westfall1993resampling}, which is particularly powerful when there is a large positive correlation between the test statistics. In \ifarxiv Appendix~\ref{sec:maxT}\fi, we derive sequential versions of such permutation based multiple testing procedures for FWER and simultaneous false discovery proportion control. These permutation based multiple testing procedures are based on the closure principle \citep{marcus1976closed} and therefore only consist of tests at level $\alpha$, which reduces the need for sequential permutation tests. Still, we believe that these can be useful in certain applications.

 \subsection*{Acknowledgments}
The authors thank Leila Wehbe for stimulating practical discussions. LF acknowledges funding by the Deutsche Forschungsgemeinschaft (DFG, German Research Foundation) – Project number 281474342/GRK2224/2. AR was funded by NSF grant DMS-2310718.

\bibliography{main}
\bibliographystyle{plainnat}

\begin{appendix}

\section{Sequential permutation based multiple testing procedures\label{sec:maxT}}

 We consider the same setting as in Section~\ref{sec:con_cal}. However, instead of defining the p-values as in \eqref{eq:cont_p-value}, we construct a level-$\alpha$ permutation test $\phi_I$ for each intersection hypothesis $H_0^I=\bigcap_{i\in i} H_0^i$, $I\subseteq \{1,\ldots,M\}$, and then use the closure principle to obtain decisions for individual hypotheses.  There are two versions of the closure principle. The initial version was proposed by \citet{marcus1976closed} who showed that the rejection set defined by 
 \begin{align}
  R=\{i\in \{1,\ldots,M\}: \phi_I=1 \text{ for all } I \text{ with } i\in I\}   
 \end{align}
controls the FWER. However, \citet{goeman2011multiple} showed that the closure principle can also be used for simultaneous true discovery control. In particular, they proved that
\begin{align*}
    \mathbb{P}(\boldsymbol{d}(S)\leq |S\cap I_1| \text{ for all } S\subseteq \{1,\ldots,M\}) \geq 1-\alpha,
\end{align*}
where $\boldsymbol{d}(S)\coloneqq \min\{|S\setminus I| : I\subseteq \mathbb{N}, \phi_I=0 \}$ and $I_1=\{1,\ldots,M\}\setminus I_0$.
In an earlier paper, \citet{genovese2006exceedance} proposed an equivalent approach without using the closure principle. 

Most permutation based multiple testing procedures are based on choosing some combination function $C_I$ for each $I\subseteq M$, defining the intersection tests $\phi_I$ by the classical permutation p-value applied on $C_I((Y_0^i)_{i\in I}), C_I((Y_1^i)_{i\in I}), \ldots, C_I((Y_B^i)_{i\in I})$ and then applying the closure principle with the intersection tests $\phi_I$. To make sure that the $\phi_I$ are indeed level-$\alpha$ tests such that this yields a valid procedure, the following assumption is usually made.


\begin{assumption} \label{assump:closure_prin}
     The vectors of test statistics corresponding to true hypotheses $(Y_0^i)_{i\in I_0}, (Y_1^i)_{i\in I_0}, \ldots,$ are jointly exchangeable.
\end{assumption}

This general approach encompasses many existing permutation based multiple testing methods.
For example, the MaxT approach for FWER control \citep{westfall1993resampling} is obtained by $C_I((Y_j^i)_{i\in I})=\max_{i\in I} Y_j^i$. \citet{vesely2023permutation} focus on procedures providing true discovery control with $C_I((Y_j^i)_{i\in I})=\sum_{i\in I} Y_j^i$. Furthermore, the method by \citet{hemerik2018false}, which uniformly improves the popular significance analysis of microarrays (SAM) procedure \citep{tusher2001significance}, is obtained by $C_I((Y_j^i)_{i\in I})=\# \{i\in I:Y_j^i \in D^i\}$ for some prespecified sets $D^i$.

It is straightforward to derive sequential versions of these multiple testing methods by replacing the classical permutation tests $\phi_I$ by the anytime-valid ones introduced in Sections~\ref{sec:old_paper} and \ref{sec:anytime_betting}. We summarize this general approach in Algorithm~\ref{alg:general_perm_based_closed_testing}.

\ifarxiv
\begin{algorithm}
\caption{Sequential permutation based multiple testing with the closure principle} \label{alg:general_perm_based_closed_testing}
 \hspace*{\algorithmicindent} \textbf{Input:} Combination functions $C_I$, $I\subseteq \mathbb{N}$, and sequences of test statistics $Y_0^i, Y_1^i, \ldots$, $i\in \{1,\ldots,M\}$.\\
 \hspace*{\algorithmicindent} 
 \textbf{Optional output:} Rejection set $R$ for FWER control or function $\boldsymbol{d}$ for simultaneous true discovery guarantee.
\begin{algorithmic}[1]
\For{$I\subseteq \{1,\ldots,M\}$} \\
Test $H_0^I$ by applying a level-$\alpha$ anytime-valid permutation test $\phi_I$ to $C_I((Y_0^i)_{i\in I}),C_I((Y_1^i)_{i\in I}),\ldots $ .
\EndFor
\State $R=\{i\in \{1,\ldots,M\}: \phi_I=1 \text{ for all } I \text{ with } i\in I\}$
\State $\boldsymbol{d}(S)=\min\{|S\setminus I|: I\subseteq \{1,\ldots,M\}, \phi_I=0\}$ for all $S\subseteq \{1,\ldots,M\}$
\State \Return $R$, $\boldsymbol{d}$
\end{algorithmic}
\end{algorithm}
\fi

For a large number of hypotheses $M$ this general approach is computationally infeasible. For this reason, short cuts have been proposed for specific choices of the combination function $C_I$ \citep{westfall1993resampling, hemerik2018false, vesely2023permutation}. In case of the MaxT approach, the entire closed test can be performed with a maximum number of $M$ intersection tests \citep{westfall1993resampling}. It might be the case that using these short cuts only works for specific betting strategies. For example, the MaxT short cut can be used with all betting strategies with nonincreasing wealth for an increasing number of losses. The detailed procedure for applying the MaxT approach with the binomial mixture strategy is provided in Algorithm~\ref{alg:MaxT}.

\ifarxiv
\begin{algorithm}
\caption{Sequential MaxT approach with the binomial mixture strategy} \label{alg:MaxT}
 \hspace*{\algorithmicindent} \textbf{Input:} Significance level $\alpha\in (0,1)$, parameter $c\in (0,\alpha)$ for the binomial mixture strategy with uniform prior $u_c$, ordered test statistics $Y_0^1\geq ... \geq Y_0^M$ and sequences of generated test statistics $Y_1^1, Y_2^1, \ldots$, $Y_1^2, Y_2^2, \ldots$, $\ldots$, $Y_1^M, Y_2^M, \ldots$ .\\
 \hspace*{\algorithmicindent} 
 \textbf{Output:} Stopping times $\tau_1,\ldots,\tau_M$ and index set of rejections $R\subseteq \{1,\ldots,M\}$. 
\begin{algorithmic}[1]
\State $\tau_i=\infty$ for all $i\in \{1,\ldots,M\}$
\State $\ell^i=0$ for all $i\in \{1,\ldots,M\}$
\State $R=\emptyset$ 
\State $A=\{1,\ldots,M\}$
\For{$t=1,2,...$}
\For{$i\in A$}
\If{$\max\{Y_t^j: j \in \{i,\ldots,M\}\} \geq \max\{Y_0^j: j \in \{i,\ldots,M\}\}$}
\State $\ell^i=\ell^i+1$ 
\If{$(1-\mathrm{Bin}(\ell^i;t+1,c))/c< \alpha$}
\State $\tau_i=\min(\tau_i,t)$
\State $\vdots$
\State $\tau_M=\min(\tau_M,t)$
\State $A=A\setminus \{i,\ldots,M\}$
\State $R=R\setminus \{i,\ldots,M\}$
\EndIf
\ElsIf{$(1-\mathrm{Bin}(\ell^i;t+1,c))/c\geq 1/\alpha$}
\State $\tau_i=t$
\State $R=R\cup \{i\}$
\State $A=A\setminus \{i\}$ 
\EndIf
\EndFor
\If{$A=\emptyset$} 
\State \Return $\tau_1,\ldots, \tau_M$, $R$
\EndIf
\EndFor
\end{algorithmic}
\end{algorithm}
\fi



\section{Detailed description of BH with the binomial mixture strategy\label{sec:detail_bm}}

 \citet{fischer2024sequential} proposed to stop for futility with their betting-based approach if the wealth drops below $\alpha$.  In the multiple testing case the level an individual hypothesis is tested at is not fixed in advance. For this reason, we propose to stop sampling for $H_0^i$ at time $t$, if $W_t^{i,\alpha_t^{\max}}<\alpha_t^{\max} $, where $\alpha_t^{\max}:=\alpha (|A_t|+m_t^*)/M$  and $A_t$ is the index set of hypotheses for which the testing process did not stop before step $t$. Hence, $\alpha_t^{\max}$ can be interpreted as the maximum level a hypothesis will be tested at by the BH procedure according to the information up to step $t$. When using the binomial mixture strategy, the combined stopping time of this stop for futility with a stop for rejection is almost surely finite (see \eqref{eq:asymp_bm}). The detailed algorithm of the entire BH procedure using the binomial mixture strategy is illustrated in Algorithm~\ref{alg:detail_BH_bm}. Of course, this is only one possible stop for futility and there are many other reasonable choices. For example, in situations where the sampling process takes several weeks or months it would also be possible to choose the stopping time interactively, meaning to revisit the data at some point and decide for which hypotheses to continue sampling based on the study interests and the evidence gathered so far. 
Our simulations in Section~\ref{sec:sim_FDR} confirm that the binomial mixture strategy with this stop for futility does not inflate the FDR if the limiting permutation p-values are PRDS.

\ifarxiv
\begin{algorithm}
\caption{BH procedure with the the binomial mixture strategy} \label{alg:detail_BH_bm}
 \hspace*{\algorithmicindent} \textbf{Input:} Significance level $\alpha\in (0,1)$, parameter $b\in (0,1)$ for the binomial mixture strategy and sequences of generated test statistics $Y_0^1, Y_1^1, \ldots$, $Y_0^2, Y_1^2, \ldots$, $\ldots$, $Y_0^M, Y_1^M, \ldots$ .\\
 \hspace*{\algorithmicindent} 
 \textbf{Output:} Stopping times $\tau_1,\ldots,\tau_M$ and index set of rejections $R\subseteq \{1,\ldots,M\}$. 
\begin{algorithmic}[1]
\State $m^*=0$ 
\State $r_{i,j}=0$ for all $i,j\in \{1,\ldots,M\}$
\State $\ell^i=0$ for all $i\in \{1,\ldots,M\}$
\State $\mathrm{crit}_i=0$ for all $i\in \{1,\ldots,M\}$
\State $A=\{1,\ldots,M\}$
\For{$t=1,2,...$}
\For{$j=1,\ldots,M$}
\State $\mathrm{crit}_j=(\mathrm{Bin})^{-1}(1-b;t+1,b\alpha j/M)-1$
\EndFor
\For{$i\in A$}
\If{$Y_t^i\geq Y_0^i$}
\State $\ell^i=\ell^i+1$ 
\EndIf
\For{$j=1,\ldots,M$} 
\If{$\ell^i\leq \mathrm{crit}_j$}
\State $r_{i,j}=1$
\EndIf
\EndFor
\EndFor
\State $m^*=\max\{m\in \{1,\ldots,M\}:\sum_{i=1}^m r_{i,m}\geq m\}$
\For{$i\in A$}
\If{$r_{i,m^*}=1$}
\State $R=R\cup \{i\}$
\State $\tau_i=t$
\State $A=A\setminus\{i\}$
\EndIf
\If{$1-\mathrm{Bin}(\ell^i;t+1,b\alpha(|A|+m^*)/M)<b[\alpha(|A|+m^*)/M]^2 $}
\State $\tau_i=t$
\State $A=A\setminus\{i\}$
\EndIf
\EndFor
\If{$A=\emptyset$} 
\State \Return $\tau_1,\ldots,\tau_M$, $R$
\EndIf
\EndFor
\end{algorithmic}
\end{algorithm}
\fi

\section{Additional notes about the RNA-seq data analysis\label{appn:rna_seq_info}}

We implemented the sequential permutation test based on the anytime-valid BC p-value and the BH procedure in an R/C++ package, \texttt{adaptiveperm}.\footnote{ \texttt{github.com/timothy-barry/adaptiveperm}} In the context of the general algorithm for sequential permutation testing (Algorithm~\ref{alg:general}), checking whether to move hypothesis $i$ from the active set $A$ into the rejection set $R$ --- i.e., line 7 --- involves performing a BH correction on all hypotheses contained within the active set $A$. To avoid performing a BH correction at every iteration of the algorithm (which would require a computationally costly sort of the p-values), we instead check whether  the maximum p-value $P^\textrm{max}$ among all hypotheses in the active set can be rejected by the BH procedure:
$$ P^\textrm{max} \coloneqq \max_{i \in A} \{ P^i \} \leq \frac{|A| m}{\alpha}.$$ If the above criterion is satisfied, we move all hypotheses from the active set into the rejection set. This procedure produces the same discovery set as the algorithm that performs a BH correction at each iteration of the loop. To see this, observe that the anytime-valid BC p-value is monotonically decreasing in time, and so if a given hypothesis can be rejected at a time $t$, then the hypothesis will also be rejected at a later time $t^* \geq t$. We implemented several other accelerations to speed the algorithm, including a high-performance Fisher-Yates sampler for permuting the data \cite{Ting2021}.

The subcutaneous tissue data, visceral adipose tissue data, and code to reproduce the RNA-seq analysis are available at the following links:
	\begin{itemize}
		\item \texttt{gtexportal.org/home/tissue/Adipose\_Subcutaneous}
		\item \texttt{gtexportal.org/home/tissue/Adipose\_Visceral\_Omentum}
		\item \texttt{github.com/timothy-barry/sequential\_perm\_testing}
	\end{itemize}

A standard step in RNA-seq analysis is to adjust for differences in library size across samples, where the library size $l_i$ of the $i$th sample is defined as the sum of the expressions across the genes in that sample:
$$l_i = \sum_{j=1}^m Y_{ij}.$$ (This step sometimes is called ``normalization.'') We normalized the data by dividing a given count $Y_{ij}$ by the library size of its sample, i.e.\ $Y_{ij}/l_i$. Li et al.\ \cite{Li2022} instead employed the normalization strategy used by the R package EdgeR \cite{Robinson2010}, which is slightly more sophisticated than the division strategy described above but similar in spirit. We filtered out genes with an expression level of zero across samples, as is standard in RNA-seq analysis \cite{Love2014}.

\section{Omitted proofs\label{appn:proofs}}

\ifarxiv
\begin{proof}[Proof of Lemma~\ref{lemma:fast_stop}]
    Note that $\pval_t^i$, $i\in \{1,\ldots,M\}$, is nonincreasing and $m_t^*$ is nondecreasing in $t$. Hence, if $\pval_t^i\leq m_t^*\alpha /M$ for some $t\in \mathbb{N}$, then $\pval_s^i\leq m_s^*\alpha /M$ for all $s \geq t$.
\end{proof}
\fi

\begin{lemma}\label{lemma:PRDS}
Let $Y$, $Z$ and $X$ be real valued random variables on $(\Omega, \mathcal{A}, \mathcal{P})$.  If $Y$ is weakly PRDS on $X$, $Z$ strongly PRDS on $Y$, and $Z$ and $X$ independent conditional on $Y$, then $Z$ is weakly PRDS on $X$.
\end{lemma}
\begin{proof}
Let $D$ be an increasing set and $x^*\geq x$, then
    \begin{align*}
        \mathbb{P}(Z \in D\mid X\leq x) &= \mathbb{E}_{Y\mid X\leq x}[\mathbb{P}(Z\in D\mid Y, X\leq x)] \\
        &= \mathbb{E}_{Y\mid X\leq x}[\mathbb{P}(Z\in D\mid Y)] \\
        &= \int_{\Omega}  \mathbb{P}(Z\in D\mid Y=y) dF_{Y\mid X\leq x}(y) \\
        &\leq \int_{\Omega}  \mathbb{P}(Z\in D\mid Y=y) dF_{Y\mid X\leq x^*}(y) \\
        &= \mathbb{P}(Z \in D\mid X\leq x^*),
    \end{align*}
    where the inequality follows from the fact that $P(Z\in D\mid Y=y)$ is increasing in $y$ and $F_{Y\mid X\leq x^*}(y)\leq F_{Y\mid X\leq x}(y)$ for all $y$.
\end{proof}

\ifarxiv
\begin{proof}[Proof of Theorem~\ref{theo:PRDS}]
    Let $x,x^*\in [0,1]$ with $x^*\geq x$, $D\subseteq [0,1]^M$ be an increasing set and $i\in I_0$ be arbitrary but fixed. In the following, we just write $\boldsymbol{\pval}$ and  $\pval^i$ instead of $\boldsymbol{\pval}^{\mathrm{avBC}}$ and $\pval^{\mathrm{avBC},i}$, respectively. We want to show that $$\mathbb{P}(
\boldsymbol{\pval}_{\boldsymbol{\gamma}(h)}^{-i}
\in D\mid\pval_{\gamma_i(h)}^i\leq x)\leq \mathbb{P}(\boldsymbol{\pval}_{\boldsymbol{\gamma}(h)}^{-i}\in D\mid \pval_{\gamma_i(h)}^i\leq x^*), $$
where $\boldsymbol{\pval}_{\boldsymbol{\gamma}(h)}^{-i}=
(\pval_{\gamma_1(h)}^1,\ldots,\pval_{\gamma_{i-1}(h)}^{i-1}, \pval_{\gamma_{i+1}(h)}^{i+1},\ldots, \pval_{\gamma_{M}(h)}^{M})$.
    The proof mainly consists of showing the following two claims.
     \begin{enumerate}
         \item $\plim^i$ is weakly PRDS on $\pval_{\gamma_i(h)}^i$.
         \item $\boldsymbol{\pval}_{\boldsymbol{\gamma}(h)}^{-i}$ is strongly PRDS on $\boldsymbol{\pval}_{\mathrm{lim}}$.
     \end{enumerate}
 Since we assumed that $\boldsymbol{\pval}_{\mathrm{lim}}$ is strongly PRDS on $I_0$, the first claim and Lemma~\ref{lemma:PRDS} imply that $\boldsymbol{\pval}_{\mathrm{lim}}$ is weakly PRDS on $\pval_{\gamma_i(h)}^i$. Together with the second claim and the fact that $\boldsymbol{\pval}_{\boldsymbol{\gamma}(h)}^{-i}$ is independent of $\pval_{\gamma_i(h)}^i$ conditional on $\boldsymbol{\pval}_{\mathrm{lim}}$ (since we sample independently for all hypotheses), the final proposition follows by Lemma~\ref{lemma:PRDS}.
     
 We start with proving the first claim. Note that $\pval_{\gamma_i(h)}^i\leq x $ iff $L_{t_x}^i\leq h-1 $,
where $t_x=\lceil h/x \rceil -1$.  In addition, De Finetti's theorem implies that $L_{t_x}^i|\plim^i=p$ follows a binomial distribution with size parameter $t_x$ and probability $p$. Hence, we obtain

\begin{align*}
    \mathbb{P}(\plim^i\leq p^* \mid \pval_{\gamma_i(h)}^i\leq x) &= \frac{\int_0^{p^*} \mathbb{P}( \pval_{\gamma_i(h)}^i\leq x \mid \plim^i= p  )\ dp }{\int_0^{1} \mathbb{P}( \pval_{\gamma_i(h)}^i\leq x \mid \plim^i= p  )\ dp } \\ 
    &= \frac{\sum_{\ell=0}^{h-1} {t_x \choose \ell} \int_0^{p^*} p^\ell (1-p)^{t_x-\ell} \ dp }{\sum_{\ell=0}^{h-1} {t_x \choose \ell} \int_0^{1} p^\ell (1-p)^{t_x-\ell} \ dp  } \\
    &= \frac{\sum_{\ell=0}^{h-1}  (1-\mathrm{Bin}(\ell;t_x+1,p^*))
   }{h
   },
\end{align*}
where $\mathrm{Bin}(\ell;t_x+1,p^*)$ is the CDF of a binomial distribution with size parameter $t_x+1$ and probability $p^*$. Since $t_x$ is decreasing in $x$, $\mathbb{P}(\plim^i\leq p^* \mid \pval_{\gamma_i(h)}^i\leq x)$ is decreasing in $x$ as well.

For the second claim, note that $\pval_{\gamma_j(h)}^j$ is independent of $\boldsymbol{\pval}_{\mathrm{lim}}^{-j} $ and $\boldsymbol{\pval}_{\boldsymbol{\gamma}(h)}^{-j} $ conditional on $\plim^j$ for all $j\in \{1,\ldots, M\}$. Therefore, it only remains to show that $\pval_{\gamma_j(h)}^j$ is strongly PRDS on $\plim^j$. Since $L_{t_u}^j|\plim^j=p$ follows a binomial distribution with probability $p$, it is immediately implied by the fact that the CDF of a binomial distribution is decreasing in its probability parameter.

\end{proof}
\fi

\ifarxiv
\begin{proof}[Proof of Theorem~\ref{theo:magic}]
   Let $m\in \{1,\ldots,M\}$ be arbitrary and set $\alpha_m=\alpha m/M$. Then, for any $i\in \{1,\ldots,M\}$,
    \begin{align*}
         \frac{h}{t+h-L_t^i} \leq \alpha_m \text{ for some } t \text{ with } L_t^i\leq h-1 
        & \ \Leftrightarrow L_t^i \leq t+h-\frac{h}{\alpha_m} \text{ for some } t \text{ with } L_t^i\leq h-1 \\
       & \ \Leftrightarrow  L_{B_{m}}^i \leq h-1 \\
       & \ \Leftrightarrow \frac{L_{B_{m}}^i+1}{B_{m}+1} \leq \alpha_m. \\
    \end{align*}
To see that $\tau_i\leq B_{|R|}$, just note that at time $t=B_{{B_{|R|}}}$ all anytime-valid BC p-values $\pval_{t}^{\mathrm{avBC},i}$, $i\in \{1,\ldots,M\}$, with $L_{t}^i\leq h-1$ reject the null hypothesis (see calculation above) and all anytime-valid BC p-values with $L_{t}^i> h-1$ have already stopped for futility.
\end{proof}
\fi

\ifarxiv
\begin{proof}[Proof of Proposition~\ref{prop:upper_bound}]
    Note that for all anytime-valid BC p-values that have not been stopped for futility yet, we have $\pval_t^{\mathrm{avBC},i}\leq h/(t+1)$. Hence, if there are at least $m$ hypotheses that have not been stopped for futility yet and $h/(t+1)\leq \alpha m/M$, we can stop and reject all remaining hypotheses. Therefore, at each time $t$, the number of hypotheses for which the sampling process has not stopped yet can be at most 
    $$
    n(t)=\max\left\{m\in \{0,\ldots,M\}:\frac{h}{t+1} > \frac{m\alpha}{M}\right\}.
    $$
    Now define 
    $$
    \tilde{n}(t)\coloneqq \begin{cases}
        M, & \frac{Mh}{\alpha (t+1)}\geq M\\
        \frac{Mh}{\alpha (t+1)}, & 1<\frac{Mh}{\alpha (t+1)}<M \\
        0, &\frac{Mh}{\alpha (t+1)}\leq 1.
    \end{cases}
    $$
    Then, $\tilde{n}(t)\geq n(t)$ for all $t$. Thus, we can write
    \begin{align*}
        \bar{\tau}&\leq  \frac{1}{M} \sum_{t\in \mathbb{N}} \tilde{n}(t)\\
        &= \frac{1}{M} \left( M \left\lfloor \frac{h}{\alpha} -1\right\rfloor +  \sum_{t=\left\lfloor \frac{h}{\alpha} \right\rfloor}^{\left\lfloor \frac{M h}{\alpha} -2 \right\rfloor} \frac{Mh}{\alpha (t+1)} \right) \\
        &= \left\lfloor \frac{h}{\alpha} -1 \right\rfloor +  \frac{h}{\alpha} \sum_{t=\left\lfloor \frac{h}{\alpha} \right\rfloor}^{\left\lfloor \frac{M h}{\alpha} -2 \right\rfloor} \frac{1}{(t+1)},
    \end{align*}
    proving the claim.
\end{proof}
\fi

\ifarxiv
\begin{proof}[Proof of Proposition~\ref{prop:futility_ind}]
    Due to De Finetti's theorem, the indicators $I_t^i=\mathbbm{1}\{Y_t^i\geq Y_0^i\}$, $t\geq 1$, for a hypothesis $H_0^i$ are i.i.d. conditional on $\plim^i$ and follow a mixture Bernoulli distribution with random probability $\plim^i$. Since $\boldsymbol{\pval}_{\mathrm{lim}}$ is independent on $I_0$,  $\boldsymbol{\pval_{\tau'}}$
 is independent on $I_0$ as well.\end{proof}
 \fi

\end{appendix}

\end{document}